	\documentclass[twoside,final]{IEEEtran}
	\usepackage{graphicx,amssymb,amsmath}
	\usepackage{multicol}
	\usepackage[noadjust]{cite}
	\usepackage{setspace}
	\usepackage{subfigure}
	\usepackage{graphicx}
	\usepackage{float}
	\usepackage {url}
	\usepackage{stfloats}
	\usepackage{amsthm}
	\usepackage{flushend}
	\usepackage{cases,subeqnarray}
	\usepackage{bm,multirow,bigstrut}
	\usepackage{amsmath, amsthm, amssymb}
	\usepackage{textcomp}
	\usepackage{latexsym,bm}
	\usepackage{booktabs}
	\usepackage{xcolor}
	\usepackage{mathtools}
	\usepackage{dsfont}
	\usepackage{extarrows}
	\usepackage{epsfig}
	\usepackage{epsfig}
	\usepackage{epstopdf}
	
	\theoremstyle{plain}

	\theoremstyle{plain}

	\newtheorem{corr}{Corollary}
	\newtheorem{prop}{Proposition}
	\usepackage{amsmath}

	\IEEEoverridecommandlockouts
	
\begin{document}
	\title{\color{black}On the Distribution of the Ratio of Products of Fisher-Snedecor $\mathcal{F}$ Random Variables and Its Applications}

	\author{Hongyang~Du, Jiayi~Zhang, Kostas~P.~Peppas,~\IEEEmembership{Senior~Member,~IEEE}, Hui Zhao, \\Bo Ai,~\IEEEmembership{Senior~Member,~IEEE}, and Xiaodan Zhang
	\thanks{H. Du and J.~Zhang are with the School of Electronic and Information Engineering, Beijing Jiaotong University, Beijing 100044, P. R. China. (e-mail: \{17211140; jiayizhang \}@bjtu.edu.cn)}
	\thanks{K. P. Peppas is with the Department of Informatics and Telecommunications, University of Peloponnese, 22131 Tripoli, Greece (e-mail: peppas@uop.gr).}
	\thanks{H. Zhao is with the Communication Systems Department, EURECOM, Sophia Antipolis 06410, France (e-mail: hui.zhao@kaust.edu.sa).}
	\thanks{B. Ai is with the  State Key Laboratory of Rail Traffic Control and Safety, Beijing Jiaotong University, Beijing 100044, China (e-mail: boai@bjtu.edu.cn).}
\thanks{X. Zhang is with School of Management, Shenzhen Institute of Information Technology, Shenzhen 518172, China. (e-mail: zhangxd@sziit.edu.cn).}
	}

	\maketitle
	\begin{abstract}
	The Fisher-Snedecor $\mathcal{F}$ distribution has been recently proposed as a more accurate and mathematically tractable composite fading model than traditional established models in some practical cases. In this paper, we firstly derive exact closed-form expressions for the main statistical characterizations of the ratio of products of $\mathcal{F}$-distributed random variables, including the probability density function, the cumulative distribution function and the moment generating function. Secondly, simple and tight  approximations to the distribution of products and ratio of products of $\mathcal{F}$-distributed random variables are presented. These analytical results can be readily employed to evaluate the performance of several emerging system configurations, including full-duplex relaying systems operating in the presence of co-channel interference and wireless communication systems enhanced with physical-layer security. The proposed mathematical analysis is substantiated by numerically evaluated results, accompanied by equivalent ones obtained using Monte Carlo simulations.
	\end{abstract}
	\begin{IEEEkeywords}
	Fisher-Snedecor $\mathcal{F}$ distribution, performance analysis, physical layer security, full-duplex relaying networks.
	\end{IEEEkeywords}
	\IEEEpeerreviewmaketitle
	\section{Introduction}
	The performance of wireless communication systems is hampered by several factors, including multipath fading and shadowing. As such, accurate channel modeling of such phenomena is crucial for purposes of accurate performance analysis and design of wireless communications \cite{wong2017key,zhang2019multiple}. In order to describe the statistics of the joint effect of multipath fading and simultaneous shadowing, various composite multipath/shadowing fading distributions have been proposed in the past decades, including the $\mathcal{K}$ \cite{abdi1998k} and the generalized-$\mathcal{K}$ distributions \cite{shankar2004error}.
	
	Recently, several generalized composite distributions have been presented in the open technical literatures which can better model the statistics of the mobile radio signal, such as the gamma shadowed Rician, $\kappa$-$\mu$, $\eta$-$\mu$, $\alpha$-$\mu$ and $\alpha$-$\kappa$-$\mu$ distributions \cite{abdi2003new,zhang2012performance,paris2014statistical,zhang2015effective,zhang2016multivariate,al2017unified,zhang2016high,ramirez2019alpha,zhang2019mixed}. Although these composite models can adequately characterize the incurred mobile signal fading phenomena, their mathematical representation is rather cumbersome or even intractable. 
	The Fisher-Snedecor $\mathcal{F}$ distribution has been recently introduced as a mathematically tractable fading model that well describes the combined effects of multipath fading and shadowing, especially in the deep fading case \cite{7886273}. This distribution can be reduced to some common fading models in some special parameter settings, such as Nakagami-$m$ and Rayleigh fading channels. Furthermore, the $\mathcal{F}$ distribution can provide a better fit to experimental data obtained at device-to-device (D2D) communications with respect to the well established generalized-$K$ distribution
	with lower computational complexity \cite{7886273}. Because of its interesting properties, the performance analysis of digital wireless communication
	systems over $\mathcal{F}$ distributed channels has been analyzed in several recent research works, e.g. see \cite{chen2018effective,badarneh2018sum,almehmadi2018effective,kong2018physical,kapucu2019analysis,zhao2019ergodic,yoo2019comprehensive,yoo2019entropy}
	and references therein.
	
	On the other hand, for the performance analysis in many practical wireless applications, the statistics of the ratios and products of fading RVs is of significant importance. For example, the performance analysis of wireless communications systems operating in the presence of co-channel interference commonly involves the evaluation of the statistics of the ratio of signals' powers, {\color{black}i.e.,} the signal-to-interference ratio (SIR). Moreover, the distribution of products of random variables is useful for the performance evaluation of multi-antenna systems operating in the presence of keyholes or relaying systems with non-regenerative relays \cite{J:PeppasCascaded}.
	
	{\color{black}The statistics of the ratio and products of fading random variables (RVs) have been extensively studied in several past research works. For example, a very generic analytical framework for the evaluation of the statistics of products and ratios of arbitrarily distributed RVs using the so-called $H$-function technique has first been proposed in \cite{CarterSpirnger}. The $H$-function distribution is a very generic statistical model that includes as special cases several well-known distributions, such as the gamma, the exponential, the Weibull and the generalized gamma ($\alpha$-$\mu$) distribution.
Several results on the distribution of the ratio of gamma, exponential, Weibull, and normal RVs have been derived in \cite{ahsen2016ratio,annavajjala2010ratio,nadarajah2006product,pham2006density}.
In a recent work \cite{8589124}, the so-called $N*$ Fisher-Snedecor $\mathcal{F}$ distribution, obtained as the products of $N$ statistically independent but not necessarily identically distributed (i.n.i.d.) $\mathcal{F}$-distributed RVs, has been introduced.
The statistical properties of the ratio of products of $\alpha$-$\mu$ RVs have been investigated in \cite{leonardo2016ratio} and \cite{matovic2013distribution}. In \cite{da2017product}, the distribution of the product of two i.n.i.d. distributed $\alpha$-$\mu$, $\kappa$-$\mu$, and $\eta$-$\mu$ RVs has been investigated}.

{\color{black} A major drawback of the above cited works, however, is that all corresponding analytical results are expressed in terms of Fox's $H$- or Meijer's $G$-functions, which are, in general, difficult to be evaluated numerically. In order to address this problem, approximate yet accurate closed-form expressions to the distribution of products and ratios have been proposed in several past research works. For example, \cite{lu2011accurate,chen2011novel}, highly accurate approximations to the distribution of products of generalized gamma and generalized normal RVs have been presented by employing a Mellin transform-based technique.}

{\color{black}On the other hand, the central limit theorem (CLT) can be efficiently used to approximate the products of RVs with a log-normal distribution. Using the CLT-based technique, in \cite{zheng2012approximation}, a log-normal approximation to the distribution of the products of Nakagami-$m$ RVs has been proposed. Additional results on the use of the log-normal distribution to the approximation of the products of Nakagami-$m$ RVs (i.e. the so-called $N$ * Nakagami fading distribution) have been obtained in \cite{Karagiannidis2015N}. The results presented in that work also revealed the connection between the log-normal distribution and the distribution of the products of RVs.}

	Motivated by the facts outlined above, in this paper we first derive closed-form expressions for the main statistics of the ratio of products of i.n.i.d. squared Fisher-Snedecor $\mathcal{F}$ RVs, including the probability density function (PDF), the cumulative distribution function (CDF) and the moment generating function (MGF).
	Then, accurate log-normal approximations to the distribution of products, ratios and ratios of products of squared $\mathcal{F}$ RVs are presented. The parameters of the log-normal distribution have been obtained utilizing the classical moment matching theory.
	In order to further highlight the usefulness of the proposed analysis, two important system configurations are assessed, namely a secure wireless communication link and a full-duplex relaying system employing the decode-and-forward (DF) protocol with co-channel interference. Extensive numerically evaluated results accompanied with Monte-Carlo simulations are further presented to validate the proposed analysis.
	
	The rest of the paper is organized as follows. In Section \ref{se:Preliminaries}, an overview of the statistical properties of the Fisher-Snedecor $\mathcal{F}$ distribution is presented. In Section \ref{Sec:Ratios}, closed-form expressions for the PDF, the CDF and the MGF of the ratio of products of squared $\mathcal{F}$-distributed RVs are derived in closed-form.
	Section \ref{Sec:APPROXIMATION} presents the proposed log-normal approximations to the ratios of products of squared $\mathcal{F}$-distributed RVs.
	In Section \ref{Sec:Applications}, the application of the proposed analysis to physical layer security and full-duplex relaying with co-channel interference is presented. Numerical and computer simulation results are presented in Section \ref{Sec:Results}, followed by Section \ref{Sec:Conclusions} concluding the paper.
	
	\emph{Notations}: $j = \sqrt { - 1} $, ${\rm{Res}}\left[ {f\left( z \right), z_k} \right]$ is the residue of the complex function $f\left( z \right)$ evaluated at $z_k$, 
	${\mathbb{E}}\left[ \cdot \right]$ denotes the expectation operator,
	$f_X\left( \cdot \right) $ denotes the PDF of the random variable (RV) $X$,
	$F_X\left( \cdot \right)$ is the CDF of $X$,
	$\mathcal{M}_X \left( \cdot \right) $ is the MGF of the RV $X$, $\Gamma \left( \cdot \right) $ is the Gamma function \cite[eq. (8.310/1)]{gradshteyn2007}, 
	${B\left( { \cdot , \cdot } \right)}$ denotes the Beta function \cite[eq. (8.384.1)]{gradshteyn2007},
	$H(\cdot)$ is the unit step function,
	${}_2{F_1}\left( { \cdot , \cdot ; \cdot ; \cdot } \right)$ denotes the Gauss hypergeometric function \cite[eq. (9.111)]{gradshteyn2007}, 
	$G \, \substack{ m , n \\ p , q}(\cdot)$ is the Meijer's G-function \cite[eq. (9.301)]{gradshteyn2007}, {\color{black}$\left\{ {{\Delta _i}} \right\}_{i = 1}^L \triangleq \{ {\Delta _1},{\Delta _2}, \cdots ,{\Delta _L}\}$ and $1 - \left\{ {{\Delta _i}} \right\}_{i = 1}^L \triangleq \{1 - {\Delta _1},1 - {\Delta _2}, \cdots ,1 - {\Delta _L} \}$.}
	
	\section{Preliminaries}\label{se:Preliminaries}
	The Fisher-Snedecor $\mathcal{F}$ distribution is a composite fading model where the received signal's small-scale variations follow a Nakagami-$m$ distribution whereas its root mean square power follows an inverse Nakagami-$m$ distribution.
	%
	A squared $\mathcal{F}$-distributed RV, $\gamma \sim \mathcal{F}\left( \bar{\gamma };m;{{m}_{s}} \right)$,
	can be mathematically obtained as the ratio of two gamma distributed RVs, $X_1$ and $X_2$, having PDFs given by
	\begin{equation}\label{Eq:PDFIx}
	f_{X_\ell}(x) = \frac{a_\ell^{a_\ell} x^{a_\ell-1}}{{b_\ell}^{b_\ell}\Gamma(a_\ell)}\exp{\left(-\frac{a_\ell x}{b_\ell}\right)},\, \forall \ell \in \{1, 2\}
	\end{equation}
	where $a_1 = m$ is the fading parameter, $a_2 = m_s$ is the shadowing parameter,  $ b_1 = \bar{\gamma }$ is the average SNR, and
	$ b_2 = m_s/(m_s-1)$.
	The PDF of $\gamma$ is given as \cite[eq. (6)]{yoo2019comprehensive} 
	\begin{equation}\label{1}
	{f_{_\gamma }}\left( \gamma  \right) = \frac{{{m^m}{{\left( {{m_s} - 1} \right)}^{{m_s}}}{{\bar \gamma }^{{m_s}}}{\gamma ^{m - 1}}}}{{B\left( {m,{m_s}} \right){{\left( {m\gamma  + \left( {{m_s} - 1} \right)\bar \gamma } \right)}^{m + {m_s}}}}},
	\end{equation}
	{\color{black}where $m>1/2$ and $m_s>1$.}
	Note that for $ m_{s}\to 0$, heavy shadowing is attained whereas shadowing vanishes as $m_s \to \infty$ ({\color{black}i.e.,} only Nakagami-$m$ small-scale fading).
	
	Finally, {\color{black}the $n_{\rm th}$ moment} of the Fisher-Snedecor $\mathcal{F}$ distribution can be	derived in closed-form as \cite[eq. (9)]{yoo2019comprehensive}
	\begin{equation}\label{nth1}
		\mathbb{E}\left[ {{\gamma ^n}} \right] = {\left( {\frac{{\left( {{m_s} - 1} \right)\bar \gamma }}{m}} \right)^n}\frac{{B\left( {m + n,{m_s} - n} \right)}}{{B\left( {m,{m_s}} \right)}}.
	\end{equation}
	
	\section{Ratio of Squared $\mathcal{F}$-distributed RVs}\label{Sec:Ratios}	
	 In this section, we derive closed-form expressions for the PDF, the CDF and the MGF of the ratio of products of independent squared $\mathcal{F}$-distributed RVs.
{\color{black}	Let us define $Z \buildrel \Delta \over = \prod\limits_{{\ell _1} = 1}^{{L_{1}}} {{\gamma _{1,{\ell _1}}}} /\prod\limits_{{\ell _2} = 1}^{{L_{2}}} {{\gamma _{2,{\ell _2}}}} $
	where ${\gamma _{1,{\ell _1}}} \sim {\cal F}\left( {{{\bar \gamma }_{1,{\ell _1}}},{m_{1,{\ell _1}}},{m_{1,{s_{{\ell _1}}}}}} \right)$ $(\ell_1=1, \cdots, L_1)$ and ${\gamma _{2,{\ell _2}}} \sim {\cal F}\left( {{{\bar \gamma }_{2,{\ell _2}}},{m_{2,{\ell _2}}},{m_{2,{s_{{\ell _2}}}}}} \right)$ $(\ell_2=1, \cdots, L_2)$ are i.n.i.d.
	$\mathcal{F}$-distributed RVs.} The following result holds.
	\begin{prop}\label{C}
	The PDF, the CDF and the MGF of $Z$ can be deduced in closed-forms as
	\begin{subequations}
	\begin{align}\label{8}
	&{{f}_{_{Z}}}\left( z \right)=\frac{1}{z{{B}_{1}}{{B}_{2}}}
	G \, \substack{ L_1+L_2 , L_1+L_2 \\ L_1+L_2 , L_1+L_2}
	\left( \frac{{{{\bar{\gamma }}}_{1}}{{A}_{2}}}{z{{{\bar{\gamma }}}_{2}}{{A}_{1}}}
	\left| \substack{ 1\!-\!\Delta_1,\,1\!-\!E_2\\
	\!\Delta_2,\, E_1 } \right.\right),
	\end{align}
	\begin{align}\label{9}
	&{{F}_{Z}}\left( z \right)\!=\!\frac{1}{{{B}_{1}}{{B}_{2}}}
	G\!\! \, \substack{ L_1+L_2+1 , L_1+L_2 \\ L_1+L_2+1 , L_1+L_2+1}
	\left( \frac{{{{\bar{\gamma }}}_{1}}{{A}_{2}}}{z{{{\bar{\gamma }}}_{2}}{{A}_{1}}}
	\left| \substack{ 1\!-\!\Delta_1,\,1\!-\!E_2,1\\
	\!\Delta_2,\, E_1,0 } \right.\right),
	\end{align}	
	\begin{align}\label{10}
	&{{\mathcal{M}}_{Z}}\!\left( s \right)\!\!=\!\frac{1}{{{B}_{1}}{{B}_{2}}}\!
	G \, \substack{ L_1+L_2 , L_1+L_2+1 \\ L_1+L_2+1 , L_1+L_2}
	\left( \frac{{{A}_{1}}}{s{{A}_{2}}}
	\left| \substack{ 1,1\!-\!\Delta_2,\,1\!-\!E_1\\
	\!\Delta_1,\, E_2 } \right.\right),
	\end{align}
	\end{subequations}
{\color{black}	where ${A_i} = \prod\limits_{{\ell _i} = 1}^{{L_i}} {\frac{{{m_{i,{\ell _i}}}}}{{{m_{i,{s_{{\ell _i}}}}} - 1}}} $, ${B_i} = \prod\limits_{{\ell _i} = 1}^{{L_i}} {\Gamma \left( {{m_{i,{\ell _i}}}} \right)\Gamma \left( {{m_{i,{s_{{\ell _i}}}}}} \right)} $, ${{\bar \gamma }_i} = \prod\limits_{{\ell _i} = 1}^{{L_i}} {{{\bar \gamma }_{i,{\ell _i}}}} $, ${\Delta _i} = \{ {m_{i,{\ell _i}}}\} _{{\ell _i} = 1}^{{L_i}}$, ${E_i} = \{ {m_{i,{s_{{\ell _i}}}}}\} _{{\ell _i} = 1}^{{L_i}}$, $\forall i \in \{1,2\}$.}
	\end{prop}
	\begin{proof}
	From the mathematical derivation of the squared $\mathcal{F}$-distribution, it can be observed that $Z$ can be expressed as the ratio of two products of $L_1+L_2$ gamma distributed RVs. In the numerator of $Z$, $L_1$ out of the $L_1+L_2$ factors are gamma distributed RVs with parameters $m_{1,\ell_1}$ and $\overline{\gamma}_{1,\ell_1}$, $\forall \ell_1 = 1,\ldots, L_1$ and the remaining $L_2$ factors are gamma distributed RVs with parameters $m_{s_{2,\ell_2}}$ and $m_{2,s_{\ell_2}}/(m_{2,s_{\ell_2}}-1)$, $\forall \ell_2 = 1,\ldots, L_2$.
	In the denominator of $Z$, $L_1$ out of the $L_1+L_2$ factors are gamma distributed RVs with parameters
	$m_{1,s_{\ell_1}}$ and $m_{1,s_{\ell_1}}/(m_{1,s_{\ell_1}}-1)$ and the remaining $L_2$ factors are gamma distributed RVs with parameters
	$m_{2,\ell_2}$ and $\overline{\gamma}_{2,\ell_2}$.
	Using \cite[eq. (4.9]{CarterSpirnger} and \cite[eq. (4.13]{CarterSpirnger}, \eqref{8} can be readily obtained after performing some straightforward algebraic manipulations.
	
	The CDF of $Z$ can be obtained as $F_Z(z) = \int_0^{z}f_z(t)\mathrm{d}t = \int_0^{\infty}H(1-t/z)f_Z(t)\mathrm{d}t$.
	By expressing the unit step function in terms of a Meijer's G-function \cite[eq. (8.4.2.1)]{book}, {\color{black}i.e.,}
	$H(1-t/z) = G \, \substack{ 1 , 0 \\ 1 , 1}
	\left( \frac{z}{t} \left| \substack{ 1\\
	0 } \right.\right)$
	and employing \cite[eq. (2.24.1.1)]{book} yields \eqref{9}.
	
	Finally, the MGF of $Z$ can be obtained as $\mathcal{M}_Z(s) = \int_0^{\infty}\exp(-s\,z)f_Z(z)\mathrm{d}z$.
	By expressing the exponential in terms of a Meijer's G-function \cite[eq. (8.4.3.1)]{book}, {\color{black}i.e.,}
	$\exp(-s\,z) = G \, \substack{ 1 , 0 \\ 0 , 1}
	\left( s\,z \left| \substack{- \\
	0 } \right.\right)$
	and employing \cite[eq. (2.24.1.1)]{book} yields \eqref{10}, thus completing the proof.
	\end{proof}
	\begin{corr}\label{A}
	Let $X \triangleq \gamma_{1} /\gamma_{2}$ where $\gamma_{1}\sim\mathcal{F}\left(\overline{\gamma}_{1}, m_{1}, m_{s_{1}}\right)$ and $\gamma_{2}\sim\mathcal{F}\left(\overline{\gamma}_{2}, m_{2}, m_{s_{2}}\right)$.
	The PDF, the CDF and the MGF of $X$ can be deduced in closed-form as
	\begin{subequations}
	{\small  \begin{align}\label{5}
		&{{f}_{_{X}}}\!\left( x \right)\!=\!\!\frac{{B\!\left( {{m_1}\! +\! {m_2},{m_{s1}}\! + \!{m_{s2}}}\! \right){x^{{m_1} - 1}}}}{{B\!\left( \!{{m_1},{m_{s1}}} \!\right)B\!\left( {{m_2},{m_{s2}}} \!\right)}}{\left( \!{\frac{{\left( {{m_{s2}} - 1}\! \right){{\bar \gamma }_{2}}{m_{1}}}}{{\left( \!{{m_{s1}} - 1} \!\right){{\bar \gamma }_{1}}{m_{2}}}}} \!\right)^{{m_1}}} \nonumber \\
		& \times\!\! {}_2{F_1}\!\!\left( \!\!{{m_1} \!+\! {m_{s1}}\!,\!{m_1}\!+\! {m_2}\!;\!\sum\limits_{\ell  = 1}^2 \!{\left( {{m_\ell }\! +\! {m_{s\ell }}} \right)} \!;\!1 \!-\! \frac{{{m_1}\!\left( \!{{m_{s2}}\! - \!1} \!\right)\!{{\bar \gamma }_2}}}{{{m_2}\!\left(\! {{m_{s1}} \!-\! 1} \!\right)\!{{\bar \gamma }_1}}}x} \!\!\right),
	\end{align}}
	{\small \begin{align}\label{6}
			& {{F}_{X}}\!\left( x \right)\!=\!\!{\left( {\frac{{\left( {{m_{s2}} - 1} \right){{\bar \gamma }_{2}}{m_{1}}}}{{\left( {{m_{s1}} - 1} \right){{\bar \gamma }_{1}}{m_{2}}}}} \right)^{{m_1}}}
			\left[\prod_{i=1}^2\Gamma(m_i)\Gamma(m_{s_i})\right]^{-1}
			\nonumber \\
			& \times G_{3,3}^{3,2}\!\left( \frac{\!{{m}_{2}}\left( {{m}_{s1}}-1 \right){{{\bar{\gamma }}}_{1}}}{x{{m}_{1}}\left( {{m}_{s2}}-1 \right){{{\bar{\gamma }}}_{2}}}\left| \begin{matrix}
				1,1-{{m}_{s2}}+{{m}_{1}},{{m}_{1}}+1  \\
				{{m}_{1}},{{m}_{1}}+{{m}_{s1}},{{m}_{1}}+{{m}_{2}}  \\
			\end{matrix} \right. \right),
		\end{align}}
		{\small \begin{align}\label{7}
				&{{\mathcal{M}}_{X}}\!\left( s \right)\!\!=\!{\left( {\frac{{\left( {{m_{s2}} - 1} \right){{\bar \gamma }_{2}}{m_{1}}}}{{\left( {{m_{s1}} - 1} \right){{\bar \gamma }_{1}}{m_{2}}}}} \right)^{{m_1}}}
				\left[\prod_{i=1}^2\Gamma(m_i)\Gamma(m_{s_i})\right]^{-1}
				\nonumber \\
				& \times G_{3,2}^{2,3}\!\left( \frac{{{m}_{1}}\left( {{m}_{s}}_{2}\!-\!1 \right){{{\bar{\gamma }}}_{2}}}{s{{m}_{2}}\left( {{m}_{s}}_{1}\!-\!1 \right){{{\bar{\gamma }}}_{1}}}\left| \begin{matrix}
					1\!-\!{{m}_{1}}\!-\!{{m}_{s1}}\!,\!1\!-\!{{m}_{1}}-{{m}_{2}}\!,1\!\!-\!\!{{m}_{1}}  \\
					0,-{{m}_{1}}\!+\!{{m}_{s2}}  \\
				\end{matrix} \right.\!\!\right)\!.
			\end{align}}
		\end{subequations}
	\end{corr}
	\begin{IEEEproof}
	Equation \eqref{5} can be obtained by employing \cite[eq. (8.4.51.1)]{book}. Equations \eqref{6} and \eqref{7} are special cases of \eqref{9} and \eqref{10}, respectively.
	\end{IEEEproof}		
	\section{New Closed-Form Approximations}\label{Sec:APPROXIMATION}
	In this section, accurate closed-form approximations to the distribution of the ratio of products of $\mathcal{F}$ RVs are presented. {\color{black} Because of the fact that the Fisher-Snedecor $\mathcal{F}$ RV can be obtained as the product of a gamma and an inverse gamma RV, ratios and products of $\mathcal{F}$ RVs can be treated as products of RVs. Motivated by the CLT-based technique mentioned in the introduction section, in this work we propose to use the log-normal distribution as an accurate closed-form approximation to the distribution of the products and ratios of $\mathcal{F}$ RVs. Furthermore, extensive experiments carried out by using the Matlab distribution fit tool have revealed that the log-normal distribution can indeed serve as an efficient approximation to a single $\mathcal{F}$ distribution as well as to the ratio of products of $\mathcal{F}$ RVs.} The parameters of the log-normal distribution are obtained using the moment matching method.
	
	\subsection{Log-Normal Approximation to a Single Squared $\mathcal{F}$-distributed RV}\label{4A}
	
	Hereafter, $Y \sim {\cal L}\left( {\mu ,\sigma } \right)$ denotes that the RV $Y$ follows the log-normal distribution with parameters $\sigma$ and $\mu$. The PDF of $Y$ is given as
	\begin{equation}
	{f_Y}\left( y \right) = \frac{1}{y}\frac{1}{{\sigma \sqrt {2\pi } }}\exp \left( { - \frac{{{{\left( {\ln y - \mu } \right)}^2}}}{{2{\sigma ^2}}}} \right).
	\end{equation}
	The $n^{t h}$ moment of $Y$ can be expressed as
	\begin{equation}\label{nth2}
		\mathbb{E}\left[Y^{n}\right]=e^{n \mu+n^{2} \sigma^{2} / 2}.
	\end{equation}
	The application of the moment matching method for the first two moments of the Fisher-Snedecor $\mathcal{F}$ distribution and the approximated log-normal distribution yields
	\begin{equation}\label{f1}
	{e^{\mu  + \frac{1}{2}{\sigma ^2}}} = \frac{{\left( {{m_s} - 1} \right)\bar \gamma }}{m}\frac{{B\left( {m + 1,{m_s} - 1} \right)}}{{B\left( {m,{m_s}} \right)}},
	\end{equation}
	\begin{equation}\label{f2}	
	{e^{2\mu  + 2{\sigma ^2}}} = {\left( {\frac{{\left( {{m_s} - 1} \right)\bar \gamma }}{m}} \right)^2}\frac{{B\left( {m + 2,{m_s} - 2} \right)}}{{B\left( {m,{m_s}} \right)}}.
	\end{equation}
	By combining \eqref{f1} and \eqref{f2}, the parameters $\sigma$ and $\mu$ can be deduced as
	\begin{equation}\label{f3}
		\left\{ \begin{array}{l}
			{\sigma ^2} = \ln \left( {{Y_f}} \right),\\
			\mu  = \ln \left( {{H_f}} \right) - \frac{1}{2}\ln \left( {{Y_f}} \right),
		\end{array} \right.
	\end{equation}
	where ${H_f} = \bar \gamma$, and ${Y_f} = \frac{{\left( {{m_s} - 1} \right)\left( {m + 1} \right)}}{{m\left( {{m_s} - 2} \right)}}$.
	
	It is noted that the expressions of distribution parameters given by \eqref{f3} may result in a relatively large approximation error, especially in the lower and upper tail regions. This is because a good fit happens only around the mean by matching the first and second moments. In order to address this issue, we adopt a modified form by considering an adjustment factor ($\varepsilon$), {\color{black}i.e.,}
	\begin{equation}
		\left\{ \begin{array}{l}
			{\sigma ^2} = \ln \left( {{Y_f} - \varepsilon } \right),\\
			\mu  = \ln \left( {{H_f}} \right) - \frac{1}{2}\ln \left( {{Y_f} - \varepsilon } \right),
		\end{array} \right.
	\end{equation}
	where $\varepsilon$ is bounded as $-{Y_f} \leq \varepsilon \leq {Y_f}$. To obtain $\varepsilon$, the numerical measure of the different (or absolute difference of two continuous distributions, {\color{black}i.e.,} Kolmogorov distance) between the exact and approximate PDFs (or CDFs) are commonly recommended.

	\subsection{Log-Normal Approximation to the Products of Two Squared $\mathcal{F}$-distributed RVs}\label{4B}
	In what follows, we propose to use the log-normal distribution as an accurate approximation to the statistics of the product of two $\mathcal{F}$-distributed RVs.
	Let $Z = X\,Y$, where
	$X \sim \mathcal{F}\left( \bar{\gamma }_x;m_x;{{m}_{s,x}} \right)$ and $Y \sim \mathcal{F}\left( \bar{\gamma }_y;m_y;{{m}_{s,y}} \right)$.
	The first and second moments of $Z$ is given by
	\begin{equation}
		\mathbb{E}\left[ Z \right] = \mathbb{E}\left[ X \right]\mathbb{E}\left[ Y \right],
	\end{equation}
	and
	\begin{equation}
		\mathbb{E}\left[ {{Z^2}} \right] = \mathbb{E}\left[ {{X^2}} \right]\mathbb{E}\left[ {{Y^2}} \right],
	\end{equation}
	respectively.
	Matching the first and second moments of $Z$ yields the following system of equations
	\begin{align}
	&{e^{\mu  + \frac{1}{2}{\sigma ^2}}} = \frac{{\left( {{m_{s,x}} - 1} \right){{\bar \gamma }_x}}}{{{m_x}}}
	\frac{{B\left( {{m_x} + 1,{m_{s,x}} - 1} \right)}}
	{{B\left( {{m_x},{m_{s,x}}} \right)}}
	\nonumber \\
	& \times \frac{{\left( {{m_{s,y}} - 1} \right){{\bar \gamma }_y}}}{{{m_y}}}\frac{{B\left( {{m_y} + 1,{m_{s,y}} - 1} \right)}}{{B\left( {{m_y},{m_{s,y}}} \right)}},
	\end{align}
	\begin{align}
	&{e^{2\mu  + 2{\sigma ^2}}} = {\left( {\frac{{\left( {{m_{s,x}} - 1} \right){{\bar \gamma }_x}}}{{{m_x}}}} \right)^2}\frac{{B\left( {{m_x} + 2,{m_{s,x}} - 2} \right)}}{{B\left( {{m_x},{m_{s,x}}} \right)}}
	\nonumber \\
	& \times{\left( {\frac{{\left( {{m_{s,y}} - 1} \right){{\bar \gamma }_y}}}{{{m_y}}}} \right)^2}\frac{{B\left( {{m_y} + 2,{m_{s,y}} - 2} \right)}}{{B\left( {{m_y},{m_{s,y}}} \right)}}.
	\end{align}
	The parameters of the approximating log-normal distribution can therefore be expressed as
\begin{equation}\label{prof1}
	\left\{ \begin{array}{l}
		{\sigma ^2} = \ln \left( {{Y_{pro}}} \right),\\
		\mu  = \ln \left( {{H_{pro}}} \right) - \frac{1}{2}\ln \left( {{Y_{pro}}} \right),
	\end{array} \right.
\end{equation}
	where
\begin{subequations}
	\begin{equation}
		{H_{pro}} = {H_{pro,1}}{H_{pro,2}} = {{\bar \gamma }_x}{{\bar \gamma }_y},
	\end{equation}
	\begin{align}
		{Y_{pro}} &= {Y_{pro,1}}{Y_{pro,2}} = \frac{{\left( {{\rm{1 + }}{m_x}} \right)\left( {{m_{s,x}} - 1} \right)}}{{{m_x}\left( {{m_{s,x}} - 2} \right)}}\nonumber \\
		&\times \frac{{\left( {{\rm{1 + }}{m_y}} \right)\left( {{m_{s,y}} - 1} \right)}}{{{m_y}\left( {{m_{s,y}} - 2} \right)}}.
	\end{align}
\end{subequations}
	
	The adjusted forms of \eqref{prof1}	can be written as
\begin{equation}\label{prof2}
	\left\{ \begin{array}{l}
		{\sigma ^2} = \ln \left( {\left( {{Y_{pro,1}} - {\varepsilon _1}} \right)\left( {{Y_{pro,2}} - {\varepsilon _2}} \right)} \right),\\
		\mu  = \ln \left( {{H_{pro,1}}{H_{pro,2}}} \right) - \frac{1}{2}\ln \left( {\left( {{Y_{pro,1}} - {\varepsilon _1}} \right)\left( {{Y_{pro,2}} - {\varepsilon _1}} \right)} \right).
	\end{array} \right.
\end{equation}
	Note that the expressions in \eqref{prof2} can be extended to the products of $N$ independent $\mathcal{F}$-distributed RVs as \eqref{prof1}, where
 \begin{subequations}
 	\begin{align}
 		{H_{pro}} = \prod\limits_{i = 1}^L {{{\bar \gamma }_i}}  = \prod\limits_{i = 1}^L {{H_{pro,i}}},
 	\end{align}
 	\begin{align}
 		{Y_{pro}} = \prod\limits_{i = 1}^L {\left( {\frac{{\left( {{\rm{1 + }}{m_i}} \right)\left( {{m_{s,i}} - 1} \right)}}{{{m_i}\left( {{m_{s,i}} - 2} \right)}} - {\varepsilon _i}} \right)}  = \prod\limits_{i = 1}^L {\left( {{Y_{pro,i}} - {\varepsilon _i}} \right)}.
 	\end{align}
 \end{subequations}
	The parameter ${\varepsilon_i}$ can be computed in a similar manner as the one proposed in Section \ref{4A}.
	
	Finally, assuming independent and identically (i.i.d.) factors with $m_x = m_y = m$, $m_{s,x} = m_{s,y} = m_s$ and
	$\overline{\gamma}_x = \overline{\gamma}_y = \overline{\gamma}$, the above expressions simplify to
	\begin{equation}
		\left\{ \begin{array}{l}
			{\sigma ^2} = N\ln \left( {\frac{{\left( {{\rm{1 + }}m} \right)\left( {{m_s} - 1} \right)}}{{m\left( {{m_s} - 2} \right)}} - \varepsilon } \right),\\
			\mu  = N\ln \left( {\overline{ \gamma}} \right) - N\frac{1}{2}\ln \left( {\frac{{\left( {{\rm{1 + }}m} \right)\left( {{m_s} - 1} \right)}}{{m\left( {{m_s} - 2} \right)}} - \varepsilon } \right).
		\end{array} \right.
	\end{equation}
	
	\subsection{Log-Normal approximation to the Ratio of Two Squared $\mathcal{F}$-distributed RVs}\label{4C}
	Hereafter, we use the moment matching method to approximate the statistics of the ratio $Z = X/Y$, where $X \sim \mathcal{F}\left( \bar{\gamma }_x;m_x;{{m}_{s,x}} \right)$ and $Y \sim \mathcal{F}\left( \bar{\gamma }_y;m_y;{{m}_{s,y}} \right)$,
	with a log-normal distribution. The first and second moments of $Z$, can be expressed as
	\begin{equation}
		\mathbb{E}\left[ Z \right] = \mathbb{E}\left[ X \right]\mathbb{E}\left[ {\frac{1}{Y}} \right],
	\end{equation}
	and
	\begin{equation}
		\mathbb{E}\left[ {{Z^2}} \right] = \mathbb{E}\left[ {{X^2}} \right]\mathbb{E}\left[ {\frac{1}{{{Y^2}}}} \right].
	\end{equation}
	The moment matching method yields estimators for $\sigma$ and $\mu$ as the solution of the following system of equations
	\begin{align}
	{e^{\mu  + \frac{1}{2}{\sigma ^2}}} &= \frac{{\left( {{m_{s,x}} - 1} \right){{\bar \gamma }_x}}}{{{m_x}}}\frac{{B\left( {{m_x} + 1,{m_{s,x}} - 1} \right)}}{{B\left( {{m_x},{m_{s,x}}} \right)}} \nonumber\\
	&\times \frac{{{m_y}}}{{\left( {{m_{s,y}} - 1} \right){{\bar \gamma }_y}}}\frac{{B\left( {{m_y} - 1,{m_{s,y}} + 1} \right)}}{{B\left( {{m_y},{m_{s,y}}} \right)}},
	\end{align}
	\begin{align}
	{e^{2\mu  + 2{\sigma ^2}}} &= {\left( {\frac{{\left( {{m_{s,x}} - 1} \right){{\bar \gamma }_x}}}{{{m_x}}}} \right)^2}\frac{{B\left( {{m_x} + 2,{m_{s,x}} - 2} \right)}}{{B\left( {{m_x},{m_{s,x}}} \right)}} \nonumber \\
	& \times {\left( {\frac{{\left( {{m_{s,y}} - 1} \right){{\bar \gamma }_y}}}{{{m_y}}}} \right)^{ - 2}}\frac{{B\left( {{m_y} - 2,{m_{s,y}} + 2} \right)}}{{B\left( {{m_y},{m_{s,y}}} \right)}}.
	\end{align}
	The solution of this system can be obtained in closed-form as
	\begin{equation}\label{ratio1}
		\left\{ \begin{array}{l}
			{\sigma ^2} = \ln \left( {{Y_{ratio}}} \right),\\
			\mu  = \ln \left( {{H_{ratio}}} \right) - \frac{1}{2}\ln \left( {{Y_{ratio}}} \right),
		\end{array} \right.
	\end{equation}
	where
\begin{subequations}
	\begin{align}
		{H_{ratio}} = \frac{{{m_{s,y}}{m_y}}}{{\left( {{m_y} - 1} \right)\left( {{m_{s,y}} - 1} \right)}}\frac{{{{\bar \gamma }_x}}}{{{{\bar \gamma }_y}}},
	\end{align}
	\begin{align}
		{Y_{ratio}} = \frac{{\left( {{m_{s,x}} - 1} \right)\left( {1 + {m_x}} \right)}}{{{m_x}\left( {{m_{s,x}} - 2} \right)}}\frac{{\left( {{m_y} - 1} \right)\left( {1 + {m_{s,y}}} \right)}}{{{m_{s,y}}\left( {{m_y} - 2} \right)}}.
	\end{align}
\end{subequations}
	
	After considering the adjustment factor	$\varepsilon$, eq. \eqref{ratio1} can be rewritten as
	\begin{equation}\label{ratio2}
		\left\{ \begin{array}{l}
			{\sigma ^2} = \ln \left( {{Y_{ratio}} - \varepsilon } \right),\\
			\mu  = \ln \left( {{H_{ratio}}} \right) - \frac{1}{2}\ln \left( {{Y_{ratio}} - \varepsilon } \right).
		\end{array} \right.
	\end{equation}
	The adjustment factor $\varepsilon$ can be obtained as before.

	\subsection{Log-Normal Approximation to the Ratio of Products of Squared $\mathcal{F}$-distributed RVs}
	In what follows, the ratio of products of squared $\mathcal{F}$-distributed RVs is approximated with the log-normal distribution using the moment matching method. {\color{black}Let $Z = \frac{X}{Y}$, where $X \triangleq \prod\limits_{{{\ell }_{1}}=1}^{{{L}_{1}}}{{{\gamma }_{1,{\ell }_{1}}}}$, $Y =\prod\limits_{{{\ell }_{2}}=1}^{{{L}_{2}}}{{{\gamma }_{2,{\ell }_{2}}}}$, $\gamma_{1,\ell_1}\sim\mathcal{F}\left(\overline{\gamma}_{1,\ell_1}, m_{1,\ell_1}, m_{1,s_{\ell_1}}\right)$ $(\ell_1=1, \cdots, L_1)$ and
		$\gamma_{2,\ell_2}\sim\mathcal{F}\left(\overline{\gamma}_{2,\ell_2}, m_{2,\ell_2}, m_{2,s_{\ell_2}}\right)$ $(\ell_2=1, \cdots, L_2)$.}
Again, we can match the first two positive moments by using an adjustable form for the parameters obtained in \eqref{4B} and \eqref{4A}.
In particular, one obtains
\begin{equation}
	\mathbb{E}\left[ Z \right] = \mathbb{E}\left[ X \right]\mathbb{E}\left[ {\frac{1}{Y}} \right]
\end{equation}
and
\begin{equation}
	\mathbb{E}\left[ {{Z^2}} \right] = \mathbb{E}\left[ {{X^2}} \right]\mathbb{E}\left[ {\frac{1}{{{Y^2}}}} \right].
\end{equation}
	The application of the moment matching method yields

	\begin{align}
	{e^{\mu  + \frac{1}{2}{\sigma ^2}}}{\rm{ = }}\prod\limits_{{\ell _1} = 1}^{{L_{\rm{1}}}} {{{\bar \gamma }_{1,{\ell _1}}}} \prod\limits_{{\ell _2} = 1}^{{L_2}} {\frac{{{m_{2,{s_{{\ell _2}}}}}{m_{2,{\ell _2}}}}}{{\left( {{m_{2,{\ell _2}}} - 1} \right)\left( {{m_{2,{s_{{\ell _2}}}}} - 1} \right)}}\frac{1}{{{{\bar \gamma }_{2,{\ell _2}}}}}} ,
	\end{align}
	\begin{align}
	{e^{2\mu  + 2{\sigma ^2}}} &= \prod\limits_{{\ell _1} = 1}^{{L_1}} {\frac{{\left( {{\rm{1 + }}{m_{1,{\ell _1}}}} \right)\left( {{m_{1,{s_{{\ell _1}}}}} - 1} \right)}}{{{m_{1,{\ell _1}}}\left( {{m_{1,{s_{{\ell _1}}}}} - 2} \right)}}}
	\nonumber \\
	&\times
	\prod\limits_{{\ell _2} = 1}^{{L_2}} {\frac{{\left( {{m_{2,{\ell _2}}} - 1} \right)\left( {1 + {m_{2,{s_{{\ell _2}}}}}} \right)}}{{{m_{2,{s_{{\ell _2}}}}}\left( {{m_{2,{\ell _2}}} - 2} \right)}}}
	\nonumber \\
	&\times \prod\limits_{i = 1}^{{L_1}} {\bar \gamma _{1,{\ell _1}}^2} \prod\limits_{{\ell _2} = 1}^{{L_2}} {{{\left( {\frac{{{m_{2,{s_{{\ell _2}}}}}{m_{2,{\ell _2}}}}}{{\left( {{m_{2,{\ell _2}}} - 1} \right)\left( {{m_{2,{s_{{\ell _2}}}}} - 1} \right)}}} \right)}^2}\frac{1}{{\bar \gamma _{2,{\ell _2}}^2}}} .
	\end{align}
	
	The above equations can be written as
	\begin{equation}\label{proratio1}
		\left\{ \begin{array}{l}
			{\sigma ^2} = \ln \left( {Y_{ratio}^{pro} } \right),\\
			\mu  = \ln \left( {H_{ratio}^{pro}} \right) - \frac{1}{2}\ln \left( {Y_{ratio}^{pro}  } \right),
		\end{array} \right.
	\end{equation}
	where
\begin{subequations}
	\begin{align}
		H_{ratio}^{pro} &= \prod\limits_{\ell_1 = 1}^{{L_{1}}} {H_{\ell_1,ratio}^{pro}} \prod\limits_{\ell_2 = 1}^{{L_2}} {H_{\ell_2,ratio}^{pro}}
		\nonumber \\
		&= \prod\limits_{\ell_1 = 1}^{{L_{1}}} {{\overline{\gamma}_{1,\ell_1}}} \prod\limits_{\ell_2 = 1}^{{L_2}} {\frac{{{m_{2,s_{\ell_2}}}{m_{2,\ell_2}}}}{{\left( {{m_{2,\ell_2}} - 1} \right)\left( {{m_{2,s_{\ell_2}}} - 1} \right)}}\frac{1}{{{\overline{\gamma}_{2,\ell_2}}}}},
	\end{align}
	\begin{align}
		& Y_{ratio}^{pro} = \prod\limits_{\ell_1 = 1}^{{L_{1}}} {Y_{\ell_1,ratio}^{pro}} \prod\limits_{\ell_2 = 1}^{{L_2}} {Y_{\ell_2,ratio}^{pro}}
		\nonumber \\
		&= \prod\limits_{\ell_1 = 1}^{{L_1}} {\frac{{\left( {{1+}{m_{1,\ell_1}}} \right)\left( {{m_{1,s_{\ell_1}}} - 1} \right)}}{{{m_{1,\ell_1}}\left( {{m_{1,s_{\ell_1}}} - 2} \right)}}} \prod\limits_{\ell_2 = 1}^{{L_2}} {\frac{{\left( {{m_{2,\ell_2}} - 1} \right)\left( {1 + {m_{2,s_{\ell_2}}}} \right)}}{{{m_{2,s_{\ell_2}}}\left( {{m_{2,\ell_2}} - 2} \right)}}}.
	\end{align}
\end{subequations}
	
	The adjusted forms of \eqref{proratio1} can be deduced as
\begin{subequations}
	\begin{align}
		{\sigma ^2} = \ln \left( {\prod\limits_{\ell_1 = 1}^{{L_{\rm{1}}}} {\left( {Y_{\ell_1,ratio}^{pro} - {\varepsilon _{1,\ell_1}}} \right)} \prod\limits_{\ell_2 = 1}^{{L_2}} {\left( {Y_{\ell_2,ratio}^{pro} - {\varepsilon _{2,\ell_2}}} \right)} } \right),
	\end{align}
	\begin{align}
		\mu  &= \ln \left( {\prod\limits_{\ell_1 = 1}^{{L_{\rm{1}}}} {H_{\ell_1,ratio}^{pro}} \prod\limits_{\ell_2 = 1}^{{L_2}} {H_{\ell_2,ratio}^{pro}} } \right)
		\nonumber \\
		&- \frac{1}{2}\ln \left( {\prod\limits_{\ell_1 = 1}^{{L_{\rm{1}}}} {\left( {Y_{\ell_1,ratio}^{pro} - {\varepsilon _{1,\ell_1}}} \right)} \prod\limits_{\ell_2 = 1}^{{L_2}} {\left( {Y_{\ell_2,ratio}^{pro} - {\varepsilon _{2,\ell_2}}} \right)} } \right).
	\end{align}
\end{subequations}
	
	For the i.i.d. case, \eqref{proratio1} can be expressed as
\begin{subequations}
	\begin{align}
		{\sigma ^2} = L\ln \left( {\frac{{\left( {{\rm{1 + }}m} \right)\left( {{m_s} - 1} \right)\left( {m - 1} \right)\left( {{m_s} + 1} \right)}}{{m{m_s}\left( {{m_s} - 2} \right)\left( {m - 2} \right)}} - \varepsilon } \right),
	\end{align}
	\begin{align}
		\mu  &= L\ln \left( {\frac{{{m_s}m}}{{\left( {{m_s} - 1} \right)\left( {m - 1} \right)}}} \right)\nonumber  \\
		& - \frac{L}{2}\ln \left( {\frac{{\left( {{\rm{1 + }}m} \right)\left( {{m_s} - 1} \right)\left( {m - 1} \right)\left( {{m_s} + 1} \right)}}{{m{m_s}\left( {{m_s} - 2} \right)\left( {m - 2} \right)}} - \varepsilon } \right).
	\end{align}
\end{subequations}

	\subsection{KS Goodness-of-fit tests}
		\begin{figure}[t]
			\centering
			\includegraphics[scale=0.65]{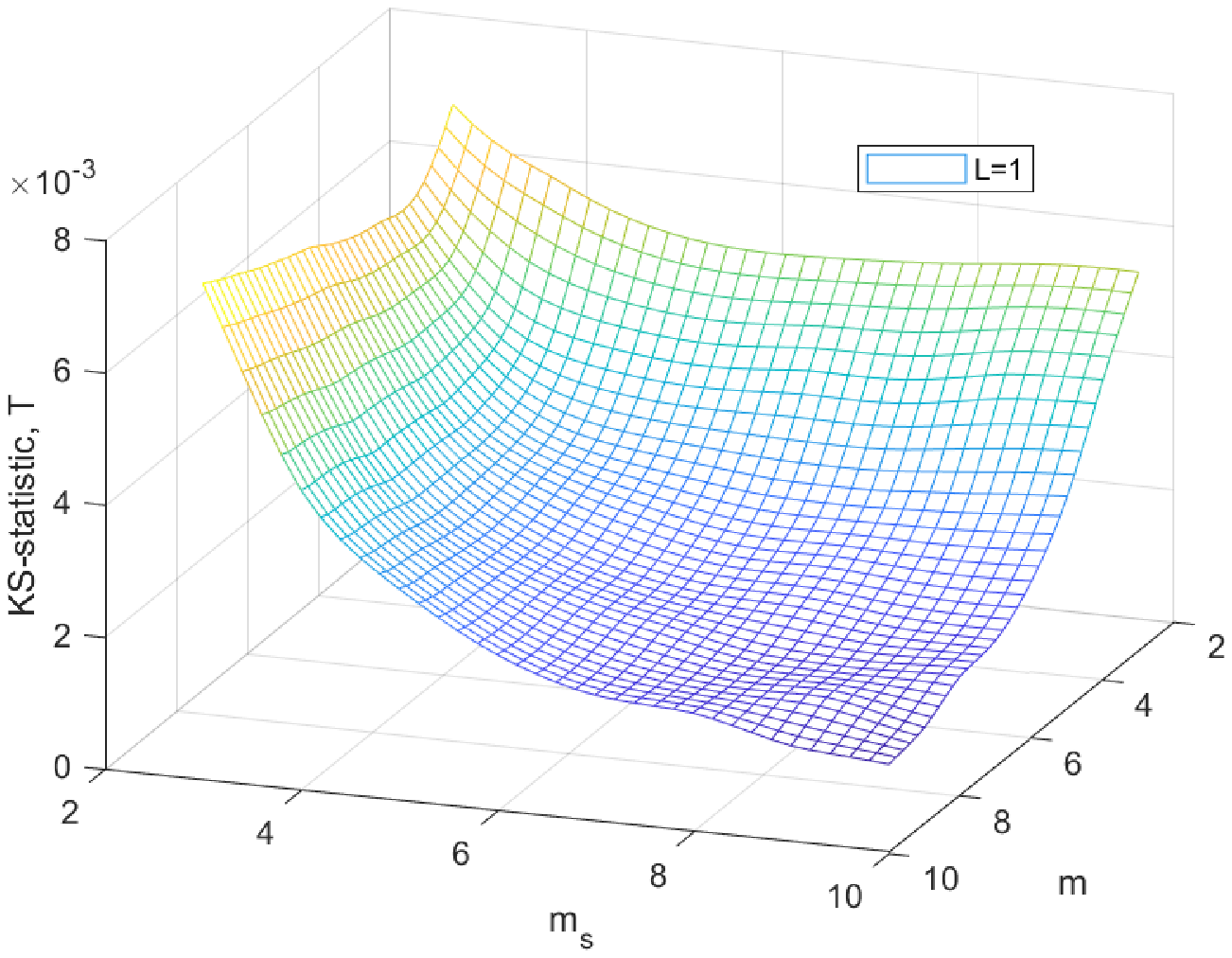}
			\caption{KS goodness-of-fit test statistic for the exact and the approximate distributions with 5$\%$ significance level for $L=1$.}
			\label{ks1}
		\end{figure}
		\begin{figure}[t]
			\centering
			\includegraphics[scale=0.65]{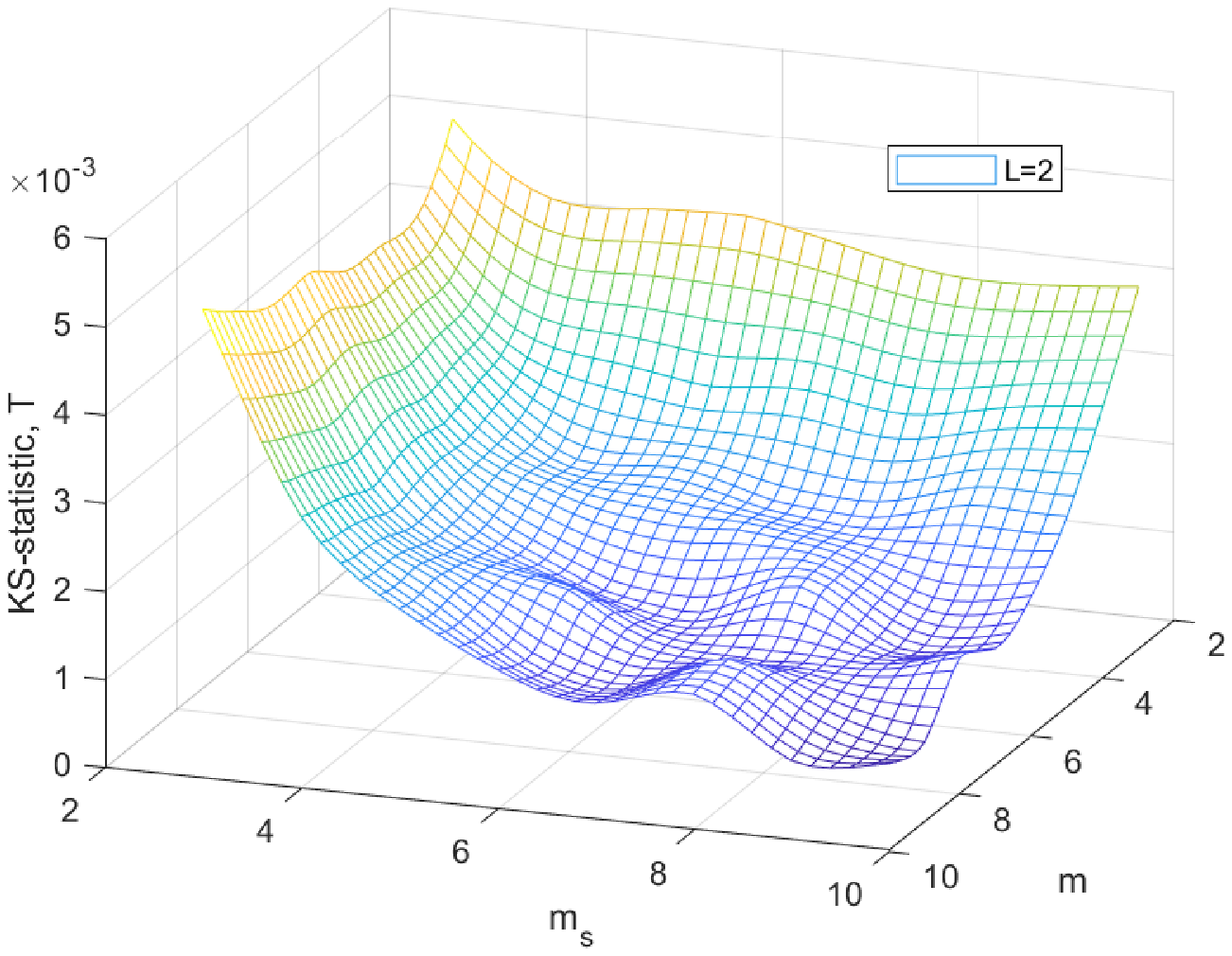}
			\caption{KS goodness-of-fit test statistic for the exact and the approximate distributions with 5$\%$ significance level for $L=2$.}
			\label{ks2}
		\end{figure}
		\begin{figure}[t]
			\centering
			\includegraphics[scale=0.65]{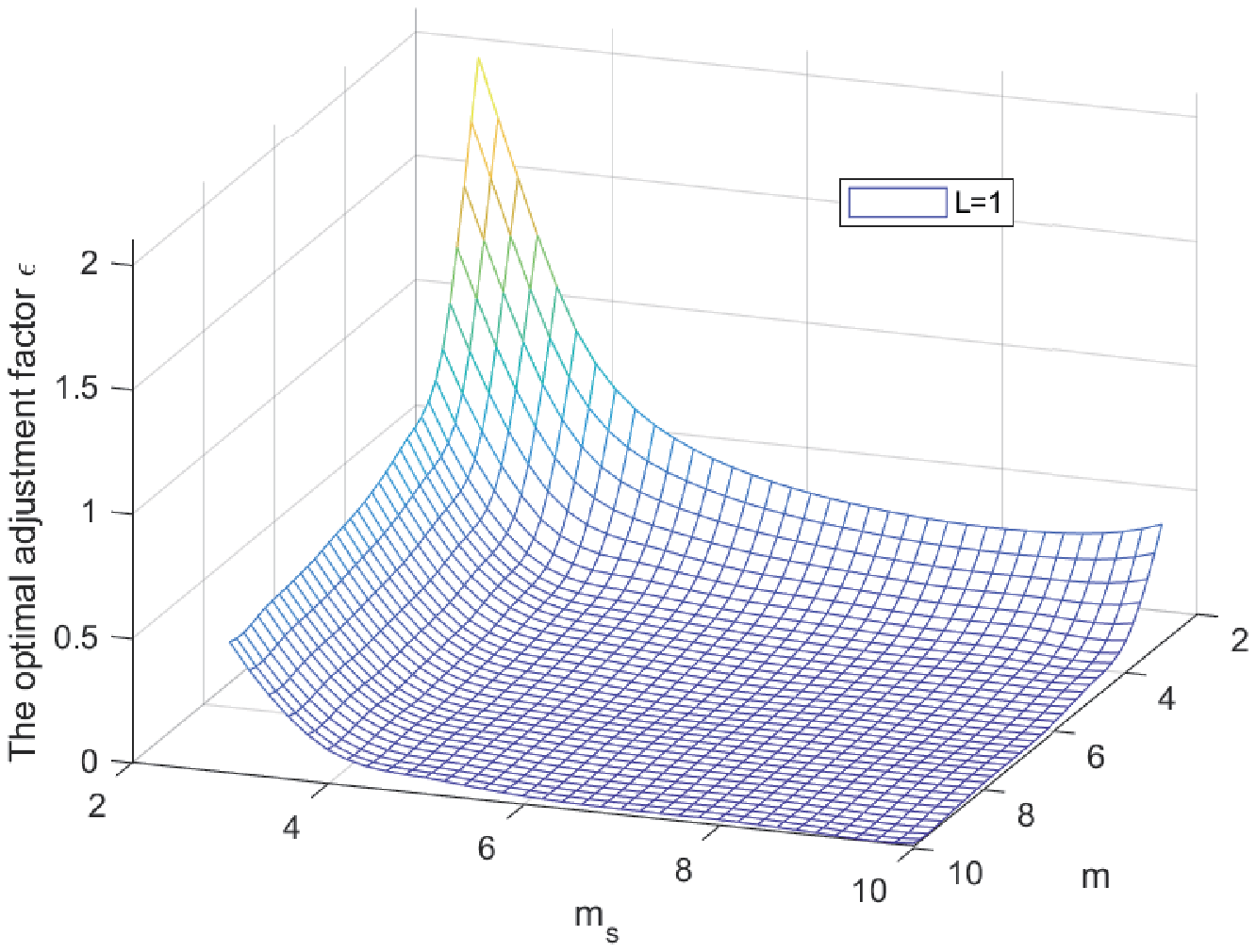}
			\caption{Adjustment factor achieving the minimum gap between the exact and the approximate log-normal distributions for $L=1$.}
			\label{ks3}
		\end{figure}
		\begin{figure}[t]
			\centering
			\includegraphics[scale=0.65]{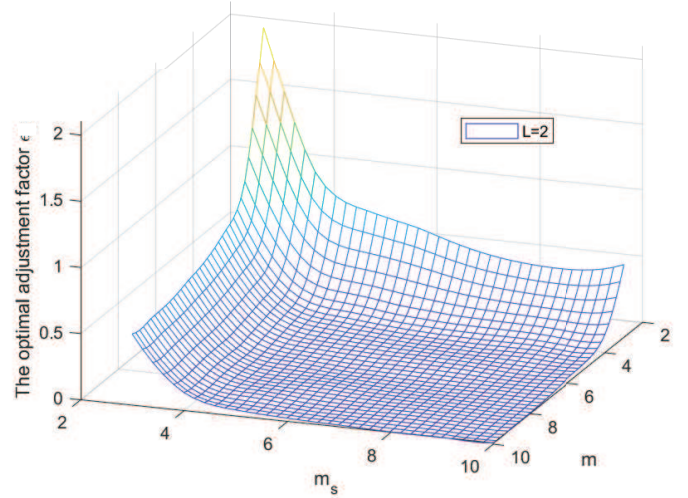}
			\caption{Adjustment factor achieving the minimum gap between the exact and the approximate log-normal distributions for $L=2$.}
			\label{ks4}
		\end{figure}
	The accuracy of the proposed approximations can be measured by using the Kolmogorov-Smirnov (KS) goodness-of-fit statistical test \cite{papoulis2002probability}. The KS test is defined as the largest absolute difference between the two cumulative distribution functions. Mathematically speaking, the KS test is expressed as
\begin{equation}
	T \triangleq \max \left|F_{S_{\gamma}}(z)-F_{\hat{S}_{\gamma}}(z)\right|,
\end{equation}	
	where $F_{S_{\gamma}}(z)$ is the empirical CDF of RV $S_{\gamma}$ and $F_{\hat{S}_{\gamma}}(z)$ is the analytical CDF of the approximate RV $\hat{S}_{\gamma}$.
	%
	In what follows, we consider the ratio of products of $L$ i.i.d RVs with parameters $m$, $m_s$ and average SNRs of 1 dB.		
	Figs. \ref{ks1} and \ref{ks2} depict the KS test statistic for different combinations of the parameters $m$, $m_s$, assuming $L=1$ and $L=2$, respectively. The exact results have been obtained by averaging at least $v=10^{4}$ random samples of $\mathcal{F}$-distributed RVs.
	Let us define $H_{0}$ as the null hypothesis under which the observed data of $S_{\gamma}$, {\color{black}i.e.,} the distribution of the ratio of products of $\mathcal{F}$-distributed RVs, belong to the CDF of the approximate distribution $F_{\hat{S}_{\gamma}}(\cdot)$, {\color{black}i.e.,} the log-normal distribution. The null hypothesis is accepted if $T<T_{\mathrm{max}}$, where $T_{\max }=\sqrt{-(1 / 2 v) \ln (\alpha / 2)}$, is the critical value which corresponds to a significance level of $\alpha$ \cite{papoulis2002probability}. For $\alpha = 5 \%$ one obtains $T_{\mathrm{max}}=0.0136$.
	Figs. \ref{ks1} and \ref{ks2} depict the test statistic, $T$, with 5$\%$ significance level for various values of $m$, $m_s$  and $L=1$ or $L=2$, respectively.
	The adjustment factor achieving the minimum gap
	between the exact and the approximated distributions is depicted in Figs. \ref{ks3} and Fig. \ref{ks4} for $L=1$ and $L=2$, respectively.
	It is clearly observed from Figs. \ref{ks1} and \ref{ks2} that the hypothesis $H_{0}$ is always accepted with 95$\%$ significance for different combinations of parameters. Hence, the log-normal distribution can be regarded as a highly accurate approximation to the exact one.
	
	\section{Applications to physical layer security and full-duplex relaying with co-channel interference}\label{Sec:Applications}
	In this section, applications of the proposed analytical framework to physical layer security and full -duplex relaying with co-channel interference are presented.
	\subsection{Exact Analysis}
	\subsubsection{Physical Layer Security}
	Physical layer security (PLS) has received considerable attention due to its ability to fortify secrecy in wireless links.
	A point-to-point channel is assumed where a transmitter (\textbf{S}) sends confidential messages to a legitimate receiver (\textbf{D}).
	This system operates in the presence of an eavesdropper (\textbf{E}).
{\color{black}According to \cite{Karagiannidis2015N},
it is assumed that nodes \textbf{S} and \textbf{D} are separated and surrounded by many stationary and moving
objects such that the signal transmitted by the source terminal
can propagate to the receiver only through electromagnetically
small apertures, i.e., keyholes among obstacles.
Each keyhole behaves like a source terminal to the next keyholes. Thus the resulting wireless channel can be modeled as a cascaded fading channel.
Moreover, it is also assumed due to high shadowing and path loss between
\textbf{S} and \textbf{D}, the direct link is highly attenuated, and thus it can be neglected. Finally, all links are subject to block i.n.i.d. $\mathcal{F}$ fading.} Next, the performance of the considered system is assessed in terms of the secrecy outage probability and the probability of non-zero secrecy capacity.
	
	The secrecy outage probability (SOP), is defined  as the probability that the instantaneous secrecy capacity falls below a target secrecy rate threshold, ${{R}_{th}}$ \cite{wyner1975wire,Bloch2008Wireless}.
	The secrecy capacity is given by \cite{wyner1975wire}
\begin{align}\label{21}
	C&=\max \left\{ {{C}_{D}}-{{C}_{E}},0 \right\}\notag\\
	& =\max \left\{ {{\log }_{2}}\left( 1+{{\gamma }_{D}} \right)-{{\log }_{2}}\left( 1+{{\gamma }_{E}} \right),0 \right\}\notag\\
	& =\left\{ \begin{matrix}
		{{\log }_{2}}\left( \frac{1+{{\gamma }_{D}}}{1+{{\gamma }_{E}}} \right),\quad \rm{if} \quad {{\gamma }_{D}}>{{\gamma }_{E}},  \\
		0,\quad \quad \quad\quad \quad\quad \rm{if}\quad{{\gamma }_{D}}\le {{\gamma }_{E}},  \\
	\end{matrix} \right.
\end{align}
	{\color{black} where ${\gamma }_{D}=\prod\limits_{{{\ell }_{D}}=1}^{{{L}_{D}}}{{{\gamma }_{D,{\ell }_{D}}}}$, ${\gamma }_{E}=\prod\limits_{{{\ell }_{E}}=1}^{{{L}_{E}}}{{{\gamma }_{E,{\ell }_{E}}}}$, ${{C}_{D}}$ and ${{C}_{E}}$ are the Shannon capacities of the main and wiretap channels.}
	Hence, SOP can be expressed as.
	{\begin{equation}
	{\rm{SOP}}\!=\!{{P}_{r}}\left\{ C\left( {{\gamma }_{D}},{{\gamma }_{E}} \right)<{{R}_{\text{th}}} \right\}\!=\!{{P}_{r}}\left\{ \frac{1\!+\!{{\gamma }_{D}}}{1\!+\!{{\gamma }_{E}}} \!<\!{{2}^{{{R}_{\text{th}}}}} \right\}.
	\end{equation}}
	A exact closed-form expression for the average SOP is, in general, difficult to be deduced in closed form. A lower bound for average SOP can be deduced as
\begin{equation}
	{\rm{SOP}}\ge {{P}_{r}}\left\{ \frac{{{\gamma }_{D}}}{{{\gamma }_{E}}}<{{2}^{{{R}_{\text{th}}}}}\right\}={{F}_{Y}}\left(\tau \right),
\end{equation}
	where $\tau \triangleq {{2}^{{{R}_{\text{th}}}}}$, {\color{black}$Y = {\gamma }_{D}/{\gamma }_{E}$} and ${{F}_{Y}}\left(\cdot \right)$ is given by \eqref{9}. As it will become evident, this bound is tight in the entire SNR region.
	
	Another fundamental performance metric of PLS is the probability of non-zero secrecy capacity (PNSC) which can be readily deduced as
\begin{align}
	{\rm{PNSC}} &= {P_r}\left\{ {{C}\left( {{\gamma _D},{\gamma _E}} \right) > 0} \right\} = {P_r}\left\{ {\left( {\frac{{1 + {\gamma _D}}}{{1 + {\gamma _E}}}} \right) > {1}} \right\}\notag\\
	& = 1 - {P_r}\left\{ {\frac{{{\gamma _D}}}{{{\gamma _E}}} \le 1} \right\} = 1 - {F_{{Y}}}\left( 1 \right).
\end{align}

	\subsubsection{Full-Duplex Relaying}
	Let us consider a two-hop full-duplex (FD) relaying network consisting of three nodes, namely one single-antenna source (\textbf{A}), one single-antenna destination (\textbf{B}) and one DF relay (\textbf{R}) equipped with one
	transmit antenna and one receive antenna to operate in full-duplex mode.

{\color{black}Throughout this analysis, it is assumed that \textbf{A}, \textbf{R} and \textbf{B} are far away from each other and surrounded by many obstacles so that the resulting communication channel can be considered as a cascaded fading channel. In addition, it is assumed that all links are subject to i.n.i.d. $\mathcal{F}$ fading.}
	By further assuming that the considered system employs the DF relaying protocol as well as an interference-limited scenario, the outage probability (OP) of the considered system can be formulated as \cite{Olivo2016An}
	
\begin{align}\label{pingle}
	{P_{\text{out}}}&= \Pr \left( {\min \left\{ {\frac{{{\gamma _{AR}}}}{{{\gamma _{RR}+1}}},{\gamma _{RB}}} \right\} < {2^R} - 1} \right)  \nonumber \\
	& \le {F_Y}\left( \sigma  \right) + {F_\gamma}\left( \sigma  \right) - {F_Y}\left( \sigma  \right){F_\gamma}\left( \sigma  \right),
\end{align}
	where ${\gamma }_{AR}=\prod\limits_{{{\ell }_{1}}=1}^{{{L}_{\text{1}}}}{{{\gamma }_{AR,{\ell }_{1}}}}$, ${\gamma }_{RB}=\prod\limits_{{{\ell }_{2}}=1}^{{{L}_{\text{2}}}}{{{\gamma }_{RB,{\ell }_{2}}}}$, $\sigma\triangleq{{2^R} - 1}$, 
	${{F}_{Y}}\left(\cdot \right)$ is the CDF of ${\gamma }_{AR}/{\gamma }_{RR}$ given by \eqref{9} and ${{F}_{\gamma}}\left(\cdot \right)$ is the CDF of $\gamma_{RB}$ given by \cite[eq. (20)]{8589124}.
	{\color{black}
	
\subsection{Asymptotic Analysis}
	Although the previously derived analytical results have been obtained in closed-form, they provide little insight as to the factors affecting system performance. As such, an asymptotic performance analysis that becomes tight at high SNR values will be presented next. Since our newly derived formulas for the SOP, PNSC and the OP of FD relaying system require the CDF of ratio of products of $\mathcal{F}$ RVs, it suffices to derive an asymptotic expression for $F_{Y}(\sigma)$, where ${{F}_{Y}}\left(\cdot \right)$ is given by \eqref{9} and $\sigma$ is a fixed value.
	
 Employing \cite[eq. (9.301)]{gradshteyn2007}, \eqref{9} can be expressed in terms of a Mellin-Barnes integral as
	\begin{align}
{F_Y}\left( \sigma  \right) =&\frac{1}{{{B_1}{B_2}}}\frac{1}{{2\pi j}}\int_{\cal L} {\frac{{\Gamma \left( -s \right)}}{{\Gamma \left( {1 - s} \right)}}\prod\limits_{{\ell _1}{=1}}^{{L_1}} {\Gamma \left( {{m_{1,{\ell _1}}} + s} \right)} } \notag\\
&\times\prod\limits_{{\ell _2}{=1}}^{{L_2}} {\Gamma \left( {{m_{2,{s_{{\ell _{2}}}}}} + s} \right)} \prod\limits_{{\ell _2}{=1}}^{{L_2}} {\Gamma \left( {{m_{2,{\ell _2}}} - s} \right)} \notag\\
&\times\prod\limits_{{\ell _1}{=1}}^{{L_1}} {\Gamma \left( {{m_{1,{s_{{\ell _1}}}}} - s} \right)} {\left( {\frac{{{{\bar \gamma }_1}{A_2}}}{{x{{\bar \gamma }_2}{A_1}}}} \right)^{ s}}ds\notag\\
=&\frac{1}{{{B_1}{B_2}}}\frac{1}{{2\pi j}}\int_{\cal L} {v(s) ds},
	\end{align}
where $\mathcal{L}$ is the Mellin Barnes contour. When ${{\bar \gamma }_1} \to \infty$, this integral can be approximated by employing \cite[Theorem 1.7]{kilbas2004h} as.
	\begin{align}
{F_Y}\left( \sigma  \right) \approx &\left\{\!\!\! \begin{array}{l}
\frac{1}{{{B_1}{B_2}}}{\rm{Res}}\left[ {v\left( s \right), - {m_{1,k}}} \right],\quad{\rm if}\quad\!\!{m_{1,k}} > {m_{2,{s_k}}}\\
\frac{1}{{{B_1}{B_2}}}{\rm{Res}}\left[ {v\left( s \right), - {m_{2,{s_k}}}} \right],\!\!\quad{\rm if}\quad\!\!{m_{1,k}} < {m_{2,{s_k}}}
\end{array} \right.\notag\\
 =& \left\{\!\!\! \begin{array}{l}
\frac{1}{{{B_1}{B_2}}}\frac{1}{{{m_{1,k}}}}\prod\limits_{{\ell _1} = 1,{\ell _1} \ne k}^{{L_1}} {\Gamma \left( {{m_{1,{\ell _1}}} - {m_{1,k}}} \right)} \\
\times \prod\limits_{{\ell _2} = 1}^{{L_2}} {\Gamma \left( {{m_{2,{s_{{\ell _2}}}}} - {m_{1,k}}} \right)} \prod\limits_{{\ell _2} = 1}^{{L_2}} {\Gamma \left( {{m_{2,{\ell _2}}} + {m_{1,k}}} \right)} \\
\times \prod\limits_{{\ell _1} = 1}^{{L_1}} {\Gamma \left( {{m_{1,{s_{{\ell _1}}}}} + {m_{1,k}}} \right)} {\left( {\frac{{{{\bar \gamma }_1}{A_2}}}{{x{{\bar \gamma }_2}{A_1}}}} \right)^{ - {m_{1,k}}}},\\
\qquad\qquad\qquad\qquad\qquad{\rm if}\quad\!\!{m_{1,k}} < {m_{2,{s_k}}}\\
\frac{1}{{{B_1}{B_2}}}\frac{1}{{{m_{2,{s_k}}}}}\prod\limits_{{\ell _1} = 1,{\ell _1} \ne k}^{{L_1}} {\Gamma \left( {{m_{1,{\ell _1}}} - {m_{2,{s_k}}}} \right)} \\
\times\!\prod\limits_{{\ell _2} = 1}^{{L_2}} {\Gamma \left( {{m_{2,{s_{{\ell _2}}}}} - {m_{2,{s_k}}}} \right)} \prod\limits_{{\ell _2} = 1}^{{L_2}} {\Gamma \left( {{m_{2,{\ell _2}}} + {m_{2,{s_k}}}} \right)} \\
\times \prod\limits_{{\ell _1} = 1}^{{L_1}} {\Gamma \left( {{m_{1,{s_{{\ell _1}}}}} + {m_{2,{s_k}}}} \right)} {\left( {\frac{{{{\bar \gamma }_1}{A_2}}}{{x{{\bar \gamma }_2}{A_1}}}} \right)^{ - {m_{2,{s_k}}}}},\\
\qquad\qquad\qquad\qquad\qquad{\rm if}\quad\!\!{m_{1,k}} < {m_{2,{s_k}}}
\end{array} \right.
	\end{align}
	where ${m_{1,k}} = \min \left\{ {{m_{1,{\ell _1}}}} \right\}$ (${\ell _1} = 1, \cdots ,{L_1}$) and ${m_{2,{s_k}}} = \min \left\{ {{m_{2,{s_{{\ell _2}}}}}} \right\}$ (${\ell _2} = 1,2, \cdots ,{L_2}$).
	
	Finally, it is worth pointing out that $\text{PNSC} \to 1$ and $\text{OP} \to {F_\gamma}\left( \sigma  \right)$ because ${F_Y}\left( \sigma  \right) \to 0$ in the high-SNR regime.
	
	}
	\section{Numerical Results}\label{Sec:Results}
	In this section, some numerical results are presented to illustrate the proposed analytical framework. All numerical results are substantiated by semi-analytical Monte Carlo simulations.
							
{\color{black}Figures \ref{fig1} and \ref{fig2} depict the PDF and the CDF of the ratio of products of $\mathcal{F}$-distributed RVs, respectively, and their proposed log-normal approximation for different values of fading parameters with $m_{1,s_{\ell_1}}=m_{1,s}$ $(\ell_1=1, \cdots, L_1)$, $m_{2,s_{\ell_2}}=m_{2,s}$ $(\ell_2=1, \cdots, L_2)$, $m_{1,{\ell_1}}=m_{1}$ and $m_{2,{\ell_2}}=m_{2}$.} As it can be observed, analytical results perfectly match Monte Carlo simulations for all considered test cases, which validate the derived analytic expressions. Furthermore, it can be observed that the approximate PDFs and CDFs, obtained using the log-normal distribution, are practically indistinguishable to the exact ones.
\begin{figure}[t]
	\centering
	\includegraphics[scale=0.65]{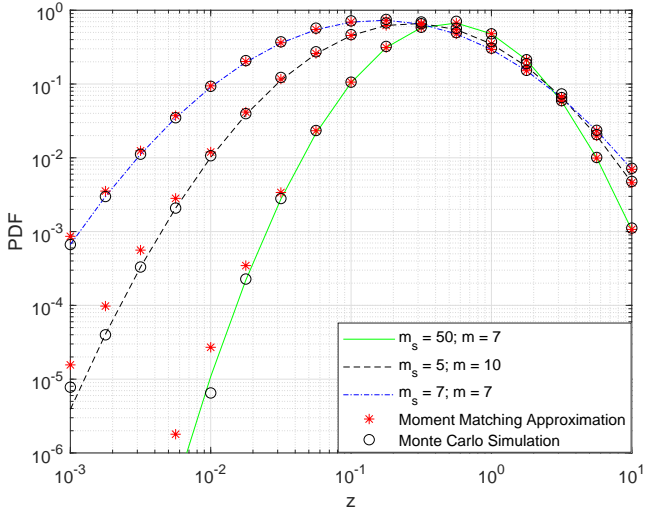}
	\caption{{\color{black}PDF of the ratio of products of squared $\mathcal{F}$-distributed RVs with $m_{s_1}=m_{s_2}=m_{s}$, $m_{1}=m_{2}=m$, $L_1=L_2=2$, $\gamma_1=\gamma_2=1$ ${\rm dB}$.}}
	\label{fig1}
\end{figure}				
\begin{figure}[t]
	\centering
	\includegraphics[scale=0.65]{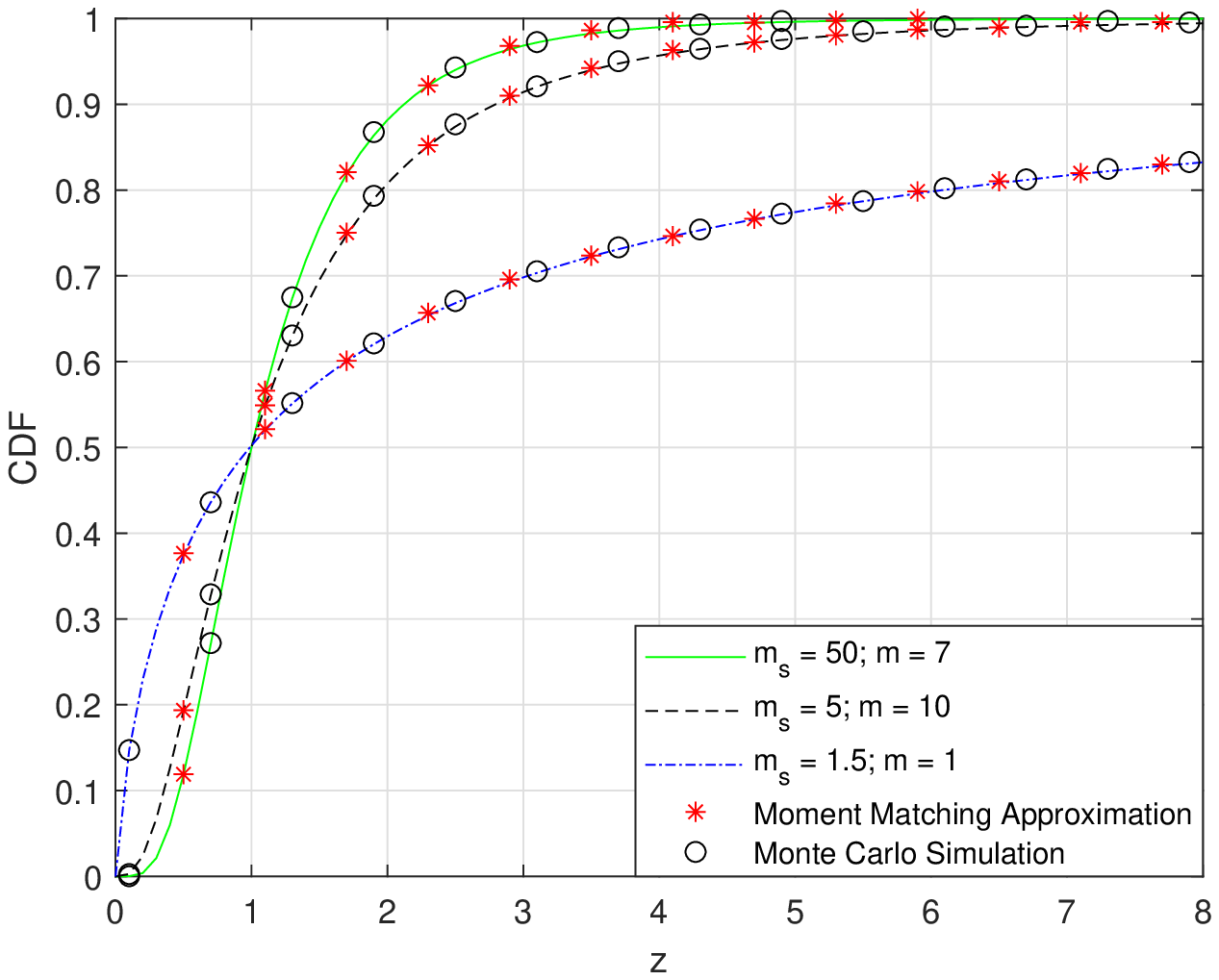}
	\caption{{\color{black}CDF of the ratio of products of squared $\mathcal{F}$-distributed RVs with $m_{1,s}=m_{2,s}=m_{s}$, $m_{1}=m_{2}=m$, $L_1=L_2=2$, $\gamma_1=\gamma_2=1$ ${\rm dB}$.}}
	\label{fig2}
\end{figure}
\begin{figure}[t]
	\centering
	\includegraphics[scale=0.65]{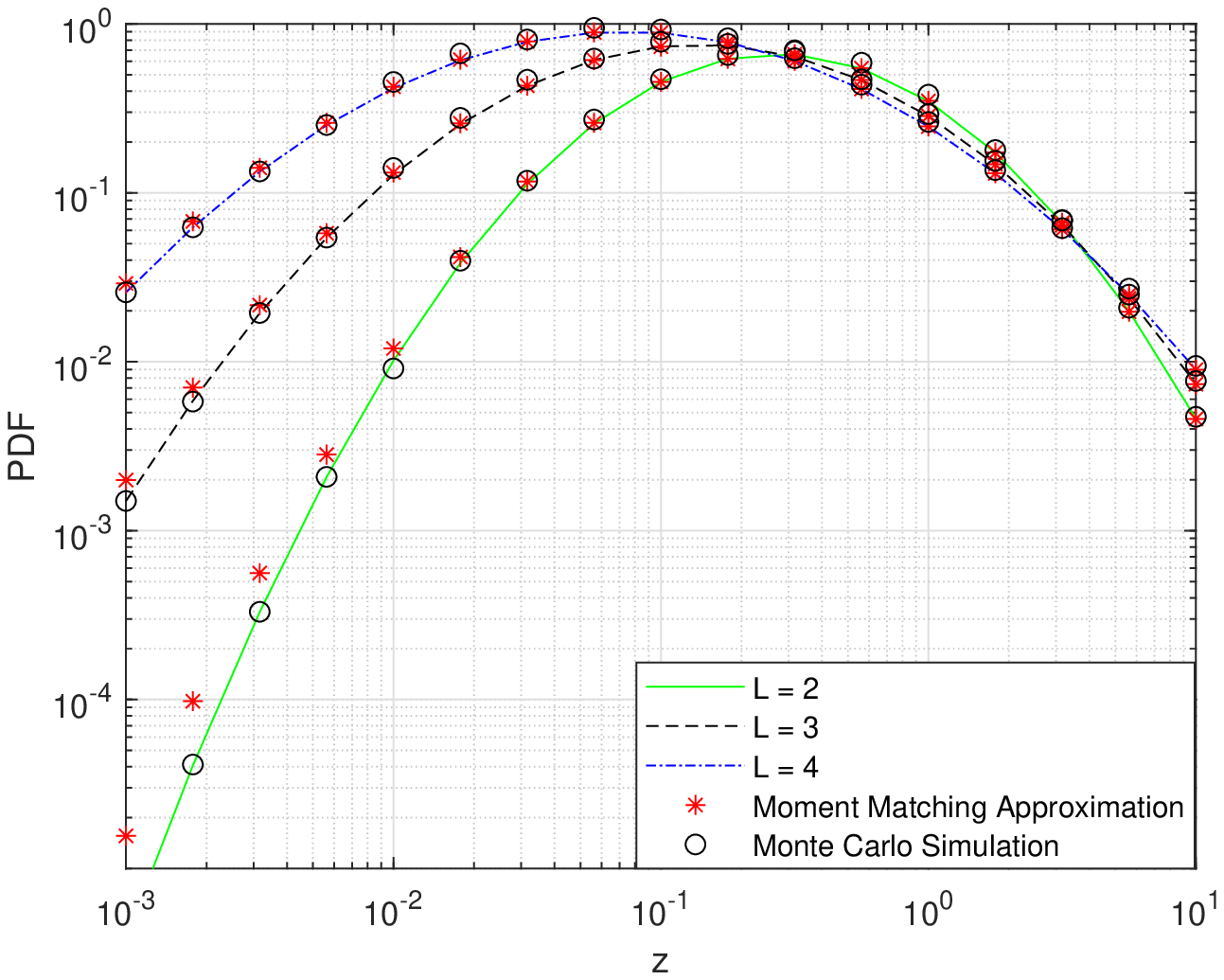}
	\caption{{\color{black}PDF of the ratio of products of squared $\mathcal{F}$-distributed RVs with $L_1=L_2=L$, $m_1=m_2=5$, $m_{1,s}=m_{2,s}=10$, $\gamma_1=\gamma_2=1$ ${\rm dB}$.}}
	\label{fig3}
\end{figure}				
\begin{figure}[t]
	\centering
	\includegraphics[scale=0.65]{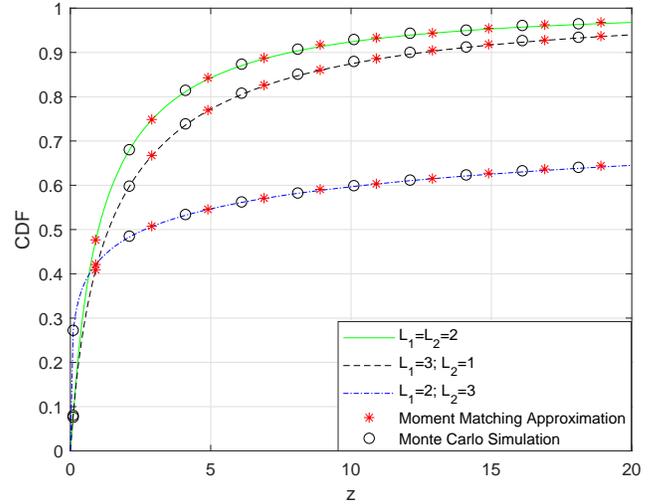}
	\caption{{\color{black}CDF of the ratio of products of squared $\mathcal{F}$-distributed RVs with $m_1=m_2=5$, $m_{1,s}=m_{2,s}=10$, $\gamma_1=\gamma_2=1$ ${\rm dB}$.}}
	\label{fig4}
\end{figure}

	%
		
	{\color{black}Figures \ref{fig3} and \ref{fig4} depict the PDF and the CDF of the ratio of products of squared $\mathcal{F}$-distributed RVs, respectively, as well as their log-normal approximation for different values of $L$ with $m_{1,s_{\ell_1}}=m_{1,s}$ $(\ell_1=1, \cdots, L_1)$, $m_{2,s_{\ell_2}}=m_{2,s}$ $(\ell_2=1, \cdots, L_2)$, $m_{1,{\ell_1}}=m_{1}$ and $m_{2,{\ell_2}}=m_{2}$}. Again, a perfect agreement between analytical and simulation results is observed, which confirms the validity of the proposed mathematical analysis. Furthermore, the approximate PDFs and CDFs obtained using the log-normal distribution  closely match the exact analytical ones.

		\begin{figure}[t]
			\centering
			\includegraphics[scale=0.65]{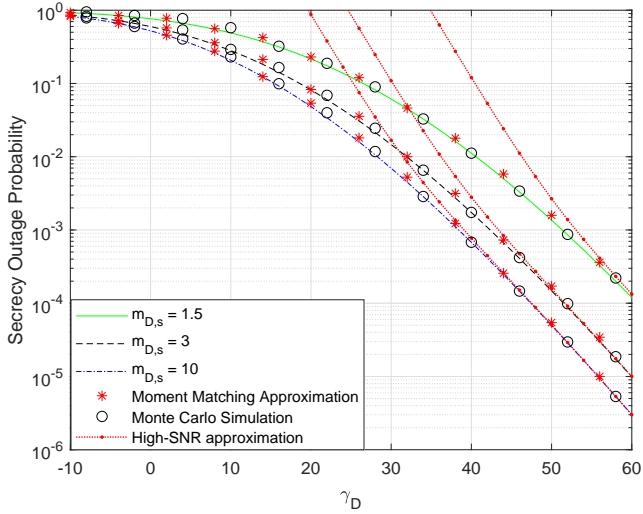}
			\caption{{\color{black}Secrecy outage probability versus $\gamma_D$ for different values of $m_{D,s}$ with $R_{\rm th}=1$, $L_D=L_E=3$, $m_D=m_E=6$, $m_{E,s}=3$ and $\gamma_E=0$ ${\rm dB}$.}}
			\label{fig5}
		\end{figure}
		\begin{figure}[t]
			\centering
			\includegraphics[scale=0.65]{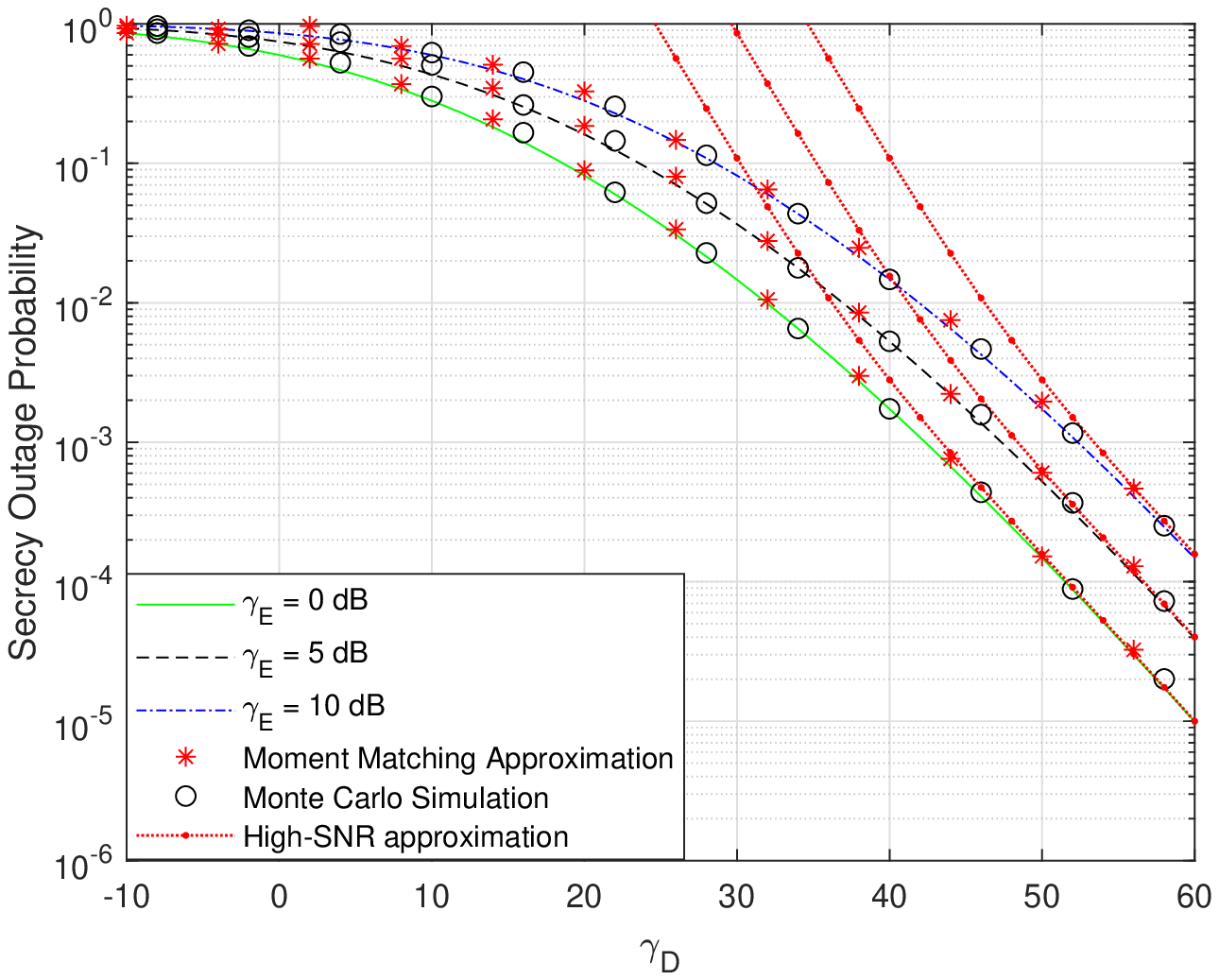}
			\caption{{\color{black}Secrecy outage probability versus $\gamma_D$ for different values of $\gamma_{E}$ with $R_{\rm th}=1$, $L_D=L_E=3$, $m_D=m_E=6$, $m_{D,s}=3$ and $m_{E,s}=3$.}}
			\label{fig6}
		\end{figure}
		\begin{figure}[t]
			\centering
			\includegraphics[scale=0.65]{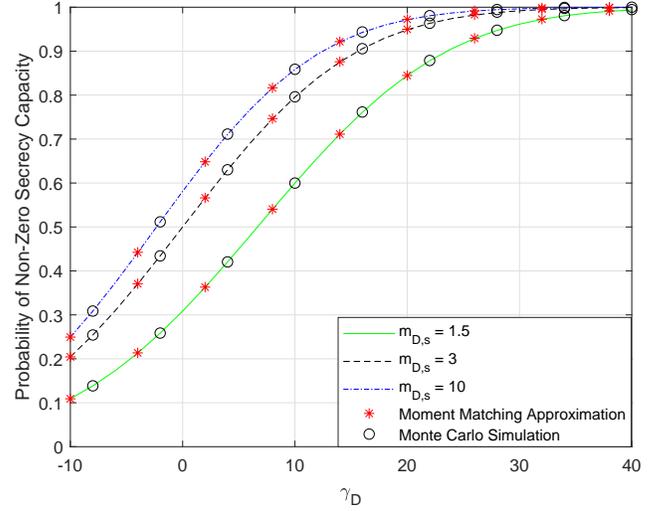}
			\caption{{\color{black}Probability of Non-Zero Secrecy Capacity versus $\gamma_D$ for different values of $m_{D,s}$ with $R_{\rm th}=1$, $L_D=L_E=3$, $m_D=m_E=6$, $m_{E,s}=3$ and  $\gamma_E=0$ ${\rm dB}$.}}
			\label{fig7}
		\end{figure}
		\setlength{\textfloatsep}{6pt}
		\begin{figure}[!t]
			\centering
			\includegraphics[scale=0.65]{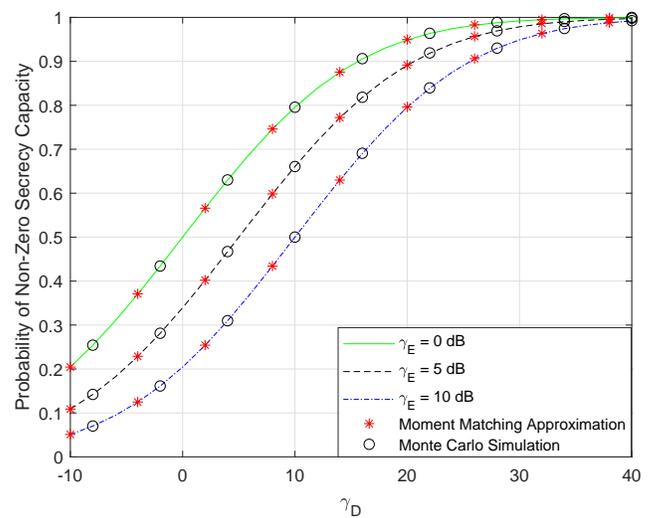}
			\caption{{\color{black}Probability of Non-Zero Secrecy Capacity versus $\gamma_D$ for different values of $\gamma_{E}$ with $R_{\rm th}=1$, $L_D=L_E=3$, $m_D=m_E=6$ and $m_{D,s}=m_{E,s}=3$.}}
			\label{fig8}
		\end{figure}
	{\color{black}Figures \ref{fig5} and \ref{fig6} depict the SOP performance of a secure wireless system operating in the presence of $\mathcal{F}$-distributed fading channels as a function of the average legitimate SNR, $\gamma_D$, with $m_{D,s_{\ell_D}}=m_{D,s}$ $(\ell_D=1, \cdots, L_D)$, $m_{E,s_{\ell_E}}=m_{E,s}$ $(\ell_E=1, \cdots, L_2)$, $m_{D,{\ell_D}}=m_{D}$ and $m_{E,{\ell_E}}=m_{E}$. Specifically, in Fig. \ref{fig5}, different values of $m_{D,s}$ have been considered and, in Fig. \ref{fig6}, different values of the average eavesdropper SNR, $\gamma_E$, have been considered.} As it is evident, the proposed lower bound is very tight and remarkably close to the exact SOP obtained by Monte Carlo simulations in the entire SNR region. Furthermore, in the same figures, the approximate SOP has been computed by using the proposed moment matching method. It is noted that in order to obtain a the best match to the exact solution, an appropriate adjustment factor for each SNR value has been taken into account, as described in Section \ref{Sec:APPROXIMATION}. {\color{black}In addition, the asymptotic expressions match well the exact ones in the high-SNR regime thus proving their validity and versatility.} Finally, it can be observed that secrecy performance deteriorates as $m_{D,s}$ decreases and $\gamma_E$ increases. This fact implies that fading conditions can be exploited to prevent the information from being overheard by an eavesdropper.

	Under the same fading conditions, Figs. \ref{fig7} and \ref{fig8} depict PNSC as a function of $\gamma_D$ for selected values of $m_{D,s}$ and $\gamma_E$, respectively. {Again, it is evident that the approximate results are very close to the exact ones and Monte Carlo simulations.} Moreover, it can be observed that larger values of $m_{D,s}$ assures secure transmission with a higher probability. This because the channel state randomness (fading) can be exploited to enhance the security.

	{\color{black}Figures \ref{fig9}-\ref{fig11} depict the outage performance of a FD relaying network as a function of the first hop average SNR, $\gamma_{AR}$, for $m_{AR,\ell_1}=m_{AR}$ $(\ell_1=1, \cdots, L_1)$, $m_{AR,s_{\ell_1}}=m_{AR,s}$, $m_{RB,\ell_2}=m_{RB}$ $(\ell_2=1, \cdots, L_2)$, $m_{RB,s_{\ell_2}}=m_{RB,s}$ and $\sigma = 1$.} For each test case, approximate results using the log-normal distribution are also depicted. Again, the adjustment factor with the best matching performance to the exact outage performance has been introduced, which is also dependent on the operating SNR. As it can be observed, analytical and approximate results agree well with Monte Carlo simulations thus confirming the correctness of the proposed analysis.
	In Fig. \ref{fig9}, different values of $m_{AR,s}$ have been considered. From the observation of Fig. \ref{fig9}, it becomes evident that the OP increases as $m_{AR,s}$ decreases, {\color{black}i.e.,} when shadowing becomes more severe. The impact of $m_{AR,s}$ on OP performance becomes more significant for large values of $\gamma_{AR}$.
	In Fig. \ref{fig10}, OP is depicted for different values of $\gamma_{RB}$. It is evident that OP increases as this $\gamma_{RB}$ decreases.
	Finally, in Fig. \ref{fig11}, OP is depicted for different values of the residual self-interference, $\gamma_{RR}$, and as it can be observed, $\gamma_{RR}$ has a detrimental impact on OP.	Thus, to exploit the potential benefits of FD relaying, self-interference mitigation matters, especially for medium-to-high SNRs. Further guidelines for self-interference mitigation on FD relay systems are available in \cite{shin2017relay,fukuzono2017self}.	{\color{black}Finally, from the observation of Figs. \ref{fig9}, \ref{fig10} and \ref{fig11} it is obvious that for large SNR values OP asymptotically becomes a constant. This constant value depends on the value of the second hop average SNR, $\gamma_{RB}$. The reason behind this behavior in the high-SNR regime is that large values of $\gamma_{AR}$ do not significantly affect ${F_Y}\left( \sigma  \right)$, as is also evident by observing eq. \eqref{pingle}. Therefore, for large SNR values, \eqref{pingle} and henceforth OP is mainly determined by $\gamma_{RB}$.}
	\begin{figure}[t]
	\centering
	\includegraphics[scale=0.65]{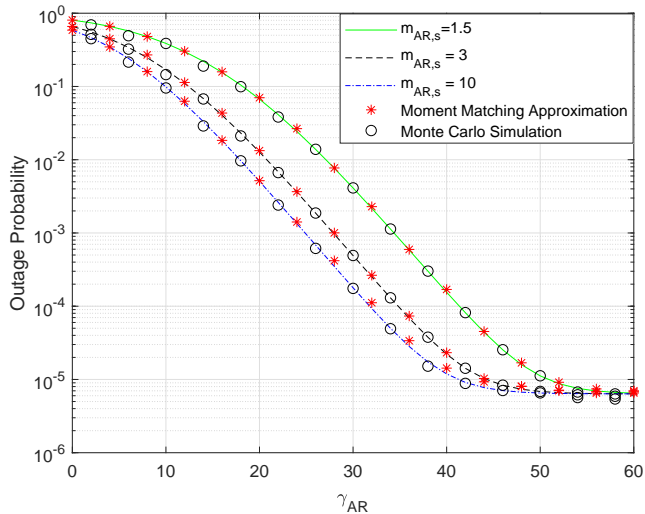}
	\caption{{\color{black}Outage probability of FD relaying network versus the average SNR of the first hop $\gamma_{AR}$ for different values of $m_{AR,s}$ with $L_1 = L_2 = 2$, $m_{AR}=6$, $m_{RR}=m_{RR,s}=1.5$, $\gamma_{RR}=5$ ${\rm dB}$, $m_{RB}=8$, $m_{RB,s}=10$, $\gamma_{RB}=15$ ${\rm dB}$.}}
	\label{fig9}
	\end{figure}
	
	\begin{figure}[t]
	\centering
	\includegraphics[scale=0.65]{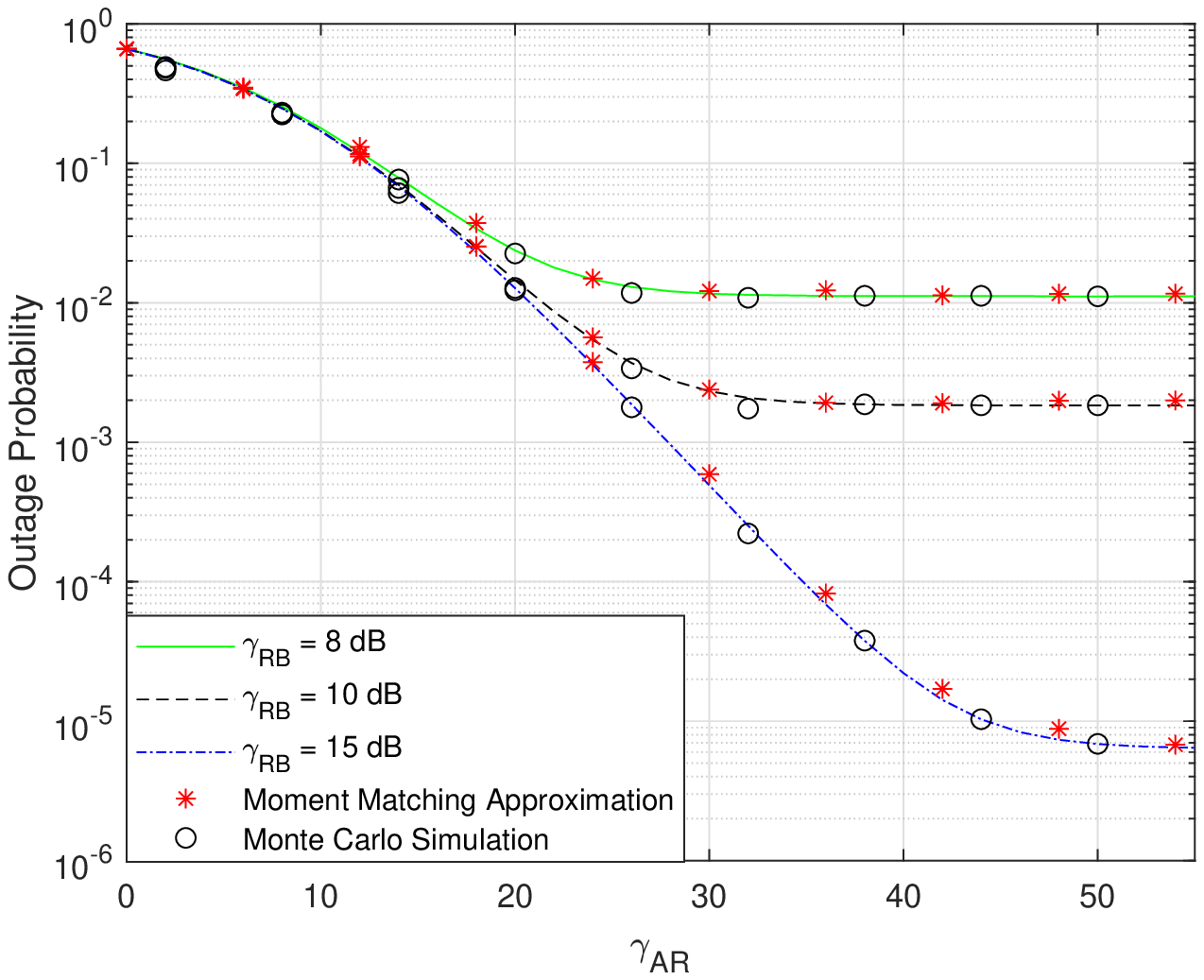}
	\caption{{\color{black}Outage probability of FD relaying network versus the average SNR of the first hop $\gamma_{AR}$ for different values of $\gamma_{RD}$ with $L_1 = L_2 = 2$, $m_{AR}=6$, $m_{RR}=m_{RR,s}=1.5$, $\gamma_{RR}=5$ ${\rm dB}$, $m_{{AR},s}=3$, $m_{RB,s}=10$, $\gamma_{RB}=15$ ${\rm dB}$.}}
	\label{fig10}
	\end{figure}
	
	\begin{figure}[t]
	\centering
	\includegraphics[scale=0.65]{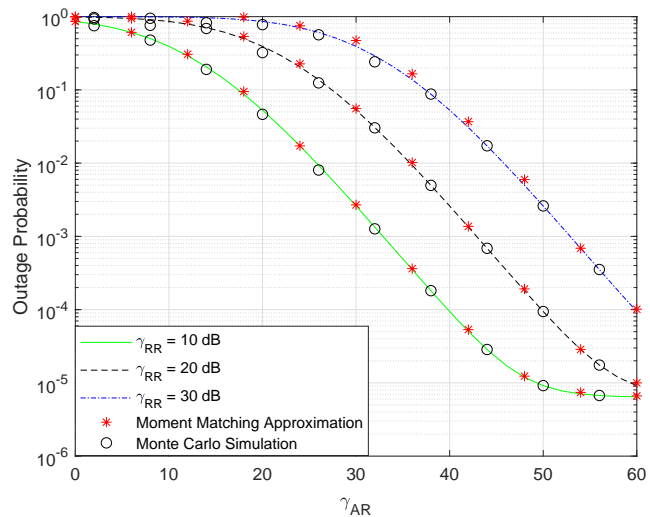}
	\caption{{\color{black}Outage probability of DF scheme versus the average SNR of the first hop $\gamma_{AR}$ for different values of $\gamma_{RR}$ with $L_1 = L_2 = 2$, $m_{AR}=6$, $m_{RR}=m_{RR,s}=1.5$, $\gamma_{RB}=15$ ${\rm dB}$, $m_{{AR},s}=3$, $m_{RB,s}=10$, $\gamma_{RB}=15$ ${\rm dB}$.}}
	\label{fig11}
	\end{figure}
	
	\section{Conclusions}\label{Sec:Conclusions}
	In this paper, exact as well as approximate closed-from expressions for the PDF, the CDF and the MGF of the ratio of products of Fisher-Snedecor $\mathcal{F}$ RVs have been derived. 
	The log-normal distribution has been selected to accurately approximate the exact PDFs and CDFs by using the moment-matching method and a properly selected adjustment factor. It has been shown that the proposed log-normal approximation is very accurate, as it yields results that are very close to the ones obtained using the exact distributions.
	The presented mathematical analysis is useful for assessing the performance of practical wireless communication systems operating in the presence of composite multipath fading and shadowing, including secure wireless communication systems and full-duplex relaying systems.
	{\color{black}The proposed framework may also be used in the study of several areas of wireless communications such as spectrum sharing, interference-limited scenarios, high resolution synthetic aperture radar clutter and multihop systems.}
	Finally, numerical results and Monte Carlo simulations have been presented to validate the proposed analysis and an excellent match has been observed.

	\bibliographystyle{IEEEtran}
	\bibliography{IEEEabrv,Ref}

\begin{thebibliography}{10}
\providecommand{\url}[1]{#1}
\csname url@samestyle\endcsname
\providecommand{\newblock}{\relax}
\providecommand{\bibinfo}[2]{#2}
\providecommand{\BIBentrySTDinterwordspacing}{\spaceskip=0pt\relax}
\providecommand{\BIBentryALTinterwordstretchfactor}{4}
\providecommand{\BIBentryALTinterwordspacing}{\spaceskip=\fontdimen2\font plus
\BIBentryALTinterwordstretchfactor\fontdimen3\font minus
  \fontdimen4\font\relax}
\providecommand{\BIBforeignlanguage}[2]{{%
\expandafter\ifx\csname l@#1\endcsname\relax
\typeout{** WARNING: IEEEtran.bst: No hyphenation pattern has been}%
\typeout{** loaded for the language `#1'. Using the pattern for}%
\typeout{** the default language instead.}%
\else
\language=\csname l@#1\endcsname
\fi
#2}}
\providecommand{\BIBdecl}{\relax}
\BIBdecl

\bibitem{wong2017key}
V.~W. Wong, R.~Schober, D.~W.~K. Ng, and L.-C. Wang, \emph{Key Technologies for
  {5G} Wireless Systems}.\hskip 1em plus 0.5em minus 0.4em\relax Cambridge
  University Press, 2017.

\bibitem{zhang2019multiple}
J.~Zhang, E.~Bj{\"o}rnson, M.~Matthaiou, D.~W.~K. Ng, H.~Yang, and D.~J. Love,
  ``Multiple antenna technologies for beyond {5G},'' \emph{arXiv:1910.00092},
  2019.

\bibitem{abdi1998k}
A.~Abdi and M.~Kaveh, ``{$K$} distribution: An appropriate substitute for
  {R}ayleigh-lognormal distribution in fading-shadowing wireless channels,''
  \emph{Electron. Lett.}, vol.~34, no.~9, pp. 851--852, Sept. 1998.

\bibitem{shankar2004error}
P.~M. Shankar, ``Error rates in generalized shadowed fading channels,''
  \emph{Wireless Personal Commun.}, vol.~28, no.~3, pp. 233--238, Mar. 2004.

\bibitem{abdi2003new}
A.~Abdi, W.~C. Lau, M.-S. Alouini, and M.~Kaveh, ``A new simple model for land
  mobile satellite channels: {F}irst-and second-order statistics,'' \emph{IEEE
  Trans. Wireless Commun.}, vol.~2, no.~3, pp. 519--528, Mar. 2003.

\bibitem{zhang2012performance}
J.~Zhang, M.~Matthaiou, Z.~Tan, and H.~Wang, ``Performance analysis of digital
  communication systems over composite $\eta$-$\mu$/gamma fading channels,''
  \emph{IEEE Trans. Veh. Technol.}, vol.~61, no.~7, pp. 3114--3124, May 2012.

\bibitem{paris2014statistical}
J.~F. Paris, ``Statistical characterization of $\kappa$-$\mu$ shadowed
  fading,'' \emph{{IEEE} Trans. Veh. Technol.}, vol.~63, no.~2, pp. 518--526,
  Feb. 2014.

\bibitem{zhang2015effective}
J.~Zhang, L.~Dai, W.~H. Gerstacker, and Z.~Wang, ``Effective capacity of
  communication systems over $\kappa$--$\mu$ shadowed fading channels,''
  \emph{Electron. Lett.}, vol.~51, no.~19, pp. 1540--1542, Sept. 2015.

\bibitem{zhang2016multivariate}
J.~Zhang, M.~Matthaiou, G.~K. Karagiannidis, and L.~Dai, ``On the multivariate
  gamma--gamma distribution with arbitrary correlation and applications in
  wireless communications,'' \emph{{IEEE} Trans. Veh. Technol.}, vol.~65,
  no.~5, pp. 3834--3840, May 2016.

\bibitem{al2017unified}
H.~Al-Hmood and H.~S. Al-Raweshidy, ``Unified modeling of composite
  $\kappa$-$\mu$/gamma, $\eta$-$\mu$/gamma, and $\alpha$-$\mu$/gamma fading
  channels using a mixture gamma distribution with applications to energy
  detection,'' \emph{{IEEE} Antennas Wireless Propag. Lett.}, vol.~16, pp.
  104--108, Apr. 2017.

\bibitem{zhang2016high}
J.~Zhang, X.~Chen, K.~P. Peppas, X.~Li, and Y.~Liu, ``On high-order capacity
  statistics of spectrum aggregation systems over $\kappa$-$\mu$ and
  $\kappa$-$\mu$ shadowed fading channels,'' \emph{IEEE Trans. Commun.},
  vol.~65, no.~2, pp. 935--944, Feb. 2016.

\bibitem{ramirez2019alpha}
P.~Ramirez-Espinosa, J.~M. Moualeu, D.~B. da~Costa, and F.~J. Lopez-Martinez,
  ``The $\alpha$-$\kappa$-$\mu$ shadowed fading distribution: {S}tatistical
  characterization and applications,'' \emph{arXiv:1904.05587}, 2019.

\bibitem{zhang2019mixed}
J.~Zhang, L.~Dai, Z.~He, B.~Ai, and O.~A. Dobre, ``Mixed-{ADC/DAC} multipair
  massive {MIMO} relaying systems: {P}erformance analysis and power
  optimization,'' \emph{IEEE Trans. Commun.}, vol.~67, no.~1, pp. 140--153,
  Jan. 2019.

\bibitem{7886273}
S.~K. Yoo, S.~L. Cotton, P.~C. Sofotasios, M.~Matthaiou, M.~Valkama, and G.~K.
  Karagiannidis, ``The {F}isher-{S}nedecor $\mathcal {F}$ distribution: A
  simple and accurate composite fading model,'' \emph{IEEE Commun. Lett.},
  vol.~21, no.~7, pp. 1661--1664, Jul. 2017.

\bibitem{chen2018effective}
S.~Chen, J.~Zhang, G.~K. Karagiannidis, and B.~Ai, ``Effective rate of {MISO}
  systems over {F}isher--{S}nedecor $\mathcal {F}$ fading channels,''
  \emph{{IEEE} Commun. Lett.}, vol.~22, no.~12, pp. 2619--2622, Dec. 2018.

\bibitem{badarneh2018sum}
O.~S. Badarneh, D.~B. da~Costa, P.~C. Sofotasios, S.~Muhaidat, and S.~L.
  Cotton, ``On the sum of {F}isher--{S}nedecor $\mathcal {F}$ variates and its
  application to maximal-ratio combining,'' vol.~7, no.~6, pp. 966--969, Dec.
  2018.

\bibitem{almehmadi2018effective}
F.~Almehmadi and O.~Badarneh, ``On the effective capacity of
  {F}isher-{S}nedecor $\mathcal {F}$ fading channels,'' \emph{Electron. Lett.},
  vol.~54, no.~18, pp. 1068--1070, Aug. 2018.

\bibitem{kong2018physical}
L.~Kong and G.~Kaddoum, ``On physical layer security over the
  {F}isher-{S}nedecor $\mathcal {F}$ wiretap fading channels,'' \emph{IEEE
  Access}, vol.~6, pp. 39\,466--39\,472, Jul. 2018.

\bibitem{kapucu2019analysis}
N.~Kapucu and M.~Bilim, ``Analysis of analytical capacity for fisher--snedecor
  $\mathcal {F}$ fading channels with different transmission schemes,''
  \emph{Electron. Lett.}, vol.~55, no.~5, pp. 283--285, May 2019.

\bibitem{zhao2019ergodic}
H.~Zhao, L.~Yang, A.~S. Salem, and M.-S. Alouini, ``Ergodic capacity under
  power adaption over fisher-snedecor $\mathcal {F}$ fading channels,''
  \emph{IEEE Commun. Lett.}, vol.~23, no. 546-549, 2019.

\bibitem{yoo2019comprehensive}
S.~K. Yoo, P.~C. Sofotasios, S.~L. Cotton, S.~Muhaidat, F.~J. Lopez-Martinez,
  J.~M. Romero-Jerez, and G.~K. Karagiannidis, ``A comprehensive analysis of
  the achievable channel capacity in $\mathcal {F}$ composite fading
  channels,'' \emph{IEEE Access}, vol.~7, pp. 34\,078 --34\,094, Mar. 2019.

\bibitem{yoo2019entropy}
S.~K. Yoo, P.~C. Sofotasios, S.~L. Cotton, S.~Muhaidat, O.~S. Badarneh, and
  G.~K. Karagiannidis, ``Entropy and energy detection-based spectrum sensing
  over $\mathcal {F}$ composite fading channels,'' \emph{IEEE Trans. Commun.},
  vol.~67, no.~7, pp. 4641--4653, Jul. 2019.

\bibitem{J:PeppasCascaded}
K.~Peppas, F.~Lazarakis, A.~Alexandridis, and K.~Dangakis, ``Cascaded
  generalised-{$K$} fading channel,'' \emph{IET Commun.}, vol.~4, no.~1, pp.
  116--124, Jan. 2010.

\bibitem{CarterSpirnger}
B.~D. Carter and M.~D. Springer, ``The distribution of products, quotients and
  powers of independent {H}--function variates,'' \emph{SIAM J. Appl. Math.},
  vol.~33, pp. 542--558, Jul. 1977.

\bibitem{ahsen2016ratio}
M.~Ahsen and S.~A. Hassan, ``On the ratio of exponential and generalized gamma
  random variables with applications to {A}d {H}oc {SISO} networks,'' in
  \emph{Proc. IEEE Veh Tchnol. Conf.}, Sep. 2016, pp. 1--5.

\bibitem{annavajjala2010ratio}
R.~Annavajjala, A.~Chockalingam, and S.~K. Mohammed, ``On a ratio of functions
  of exponential random variables and some applications,'' \emph{IEEE Trans.
  Commun.}, vol.~58, no.~11, pp. 3091--3097, Nov. 2010.

\bibitem{nadarajah2006product}
S.~Nadarajah and S.~Kotz, ``On the product and ratio of {G}amma and {W}eibull
  random variables,'' \emph{Econometric Theory}, vol.~22, no.~2, pp. 338--344,
  Feb. 2006.

\bibitem{pham2006density}
T.~Pham-Gia, N.~Turkkan, and E.~Marchand, ``Density of the ratio of two normal
  random variables and applications,'' \emph{Commun. Stat. Theory Methods},
  vol.~35, no.~9, pp. 1569--1591, Sep. 2006.

\bibitem{8589124}
O.~S. {Badarneh}, S.~{Muhaidat}, P.~C. {Sofotasios}, S.~L. {Cotton},
  K.~{Rabie}, and D.~B. {da Costa}, ``The {N}*{F}isher-snedecor $\mathcal{F}$
  cascaded fading model,'' in \emph{Proc. WiMob}, Oct. 2018, pp. 1--7.

\bibitem{leonardo2016ratio}
E.~J. Leonardo, M.~D. Yacoub, and R.~A. de~Souza, ``Ratio of products of
  $\alpha$-$\mu$ variates,'' \emph{IEEE Commun. Lett.}, vol.~20, no.~5, pp.
  1022--1025, Mar. 2016.

\bibitem{matovic2013distribution}
A.~Matovi{\'c}, E.~Meki{\'c}, N.~Sekulovi{\'c}, M.~Stefanovi{\'c},
  M.~Matovi{\'c}, and {\v{C}}.~Stefanovi{\'c}, ``The distribution of the ratio
  of the products of two independent $\alpha$-$\mu$ variates and its
  application in the performance analysis of relaying communication systems,''
  \emph{Math. Probl. Eng.}, vol. 2013, Feb. 2013.

\bibitem{da2017product}
C.~R.~N. da~Silva, E.~J. Leonardo, and M.~D. Yacoub, ``Product of two envelopes
  taken from $\alpha$-$\mu$, $\kappa$-$\mu$, and $\eta$-$\mu$ distributions,''
  \emph{IEEE Trans. Commun.}, vol.~66, no.~3, pp. 1284--1295, Nov. 2017.

\bibitem{lu2011accurate}
H.~Lu, Y.~Chen, and N.~Cao, ``Accurate approximation to the {PDF} of the
  product of independent {R}ayleigh random variables,'' \emph{IEEE Antennas
  Wireless Propagat. Lett.}, vol.~10, pp. 1019--1022, Sept. 2011.

\bibitem{chen2011novel}
Y.~Chen, G.~K. Karagiannidis, H.~Lu, and N.~Cao, ``Novel approximations to the
  statistics of products of independent random variables and their applications
  in wireless communications,'' \emph{IEEE Trans. Veh. Technol.}, vol.~61,
  no.~2, pp. 443--454, Feb. 2011.

\bibitem{zheng2012approximation}
Z.~Zheng, L.~Wei, J.~Hamalainen, and O.~Tirkkonen, ``Approximation to
  distribution of product of random variables using orthogonal polynomials for
  lognormal density,'' \emph{IEEE Commun. Lett.}, vol.~16, no.~12, pp.
  2028--2031, Dec. 2012.

\bibitem{Karagiannidis2015N}
G.~K. Karagiannidis, N.~C. Sagias, and P.~T. Mathiopoulos, ``${N} *$
  {N}akagami: A novel stochastic model for cascaded fading channels,''
  \emph{{IEEE} Trans. Commun.}, vol.~55, no.~8, pp. 1453--1458, Aug. 2007.

\bibitem{gradshteyn2007}
I.~S. Gradshteyn and I.~M. Ryzhik, \emph{Table of integrals, series, and
  products}, 7th~ed.\hskip 1em plus 0.5em minus 0.4em\relax Academic Press,
  2007.

\bibitem{book}
A.~P. Prudnikov, Y.~A. Brychkov, and O.~I. Marichev, \emph{Integrals and
  {S}eries {V}olume 3: {M}ore {S}pecial {F}unctions}, 1st~ed.\hskip 1em plus
  0.5em minus 0.4em\relax Gordon and Breach Science Publishers, 1986.

\bibitem{papoulis2002probability}
A.~Papoulis and S.~U. Pillai, \emph{Probability, random variables, and
  stochastic processes}.\hskip 1em plus 0.5em minus 0.4em\relax McGraw-Hill
  Education, 2002.

\bibitem{wyner1975wire}
A.~D. Wyner, ``The wire-tap channel,'' \emph{Bell Syst. Tech. J.}, vol.~54,
  no.~8, pp. 1355--1387, 1975.

\bibitem{Bloch2008Wireless}
M.~Bloch, J.~Barros, M.~R.~D. Rodrigues, and S.~W. Mclaughlin, ``Wireless
  information-theoretic security,'' \emph{IEEE Trans. Inf. Theory}, vol.~54,
  no.~6, pp. 2515--2534, Jun. 2008.

\bibitem{Olivo2016An}
E.~E.~B. Olivo, D.~P.~M. Osorio, H.~Alves, J.~C.~S. Filho, and M.~Latva-Aho,
  ``An adaptive transmission scheme for cognitive decode-and-forward relaying
  networks: Half duplex, full duplex, or no cooperation,'' \emph{IEEE Trans.
  Wireless Commun.}, vol.~15, no.~8, pp. 5586--5602, Aug. 2016.

\bibitem{kilbas2004h}
A.~A. Kilbas, \emph{H-transforms: Theory and Applications}.\hskip 1em plus
  0.5em minus 0.4em\relax CRC Press, 2004.

\bibitem{shin2017relay}
W.~Shin, N.~Lee, H.~Yang, and J.~Lee, ``Relay-aided successive aligned
  interference cancellation for wireless $ x $ networks with full-duplex
  relays,'' \emph{IEEE Trans. Veh. Technol.}, vol.~66, no.~1, pp. 421--432,
  Jan. 2017.

\bibitem{fukuzono2017self}
H.~Fukuzono, T.~Shinagawa, M.~Yoshioka, and H.~Nakamura, ``Self-interference
  cancellation with optimal feedback path selection on full-duplex relay
  systems,'' in \emph{2017 IEEE 28th Annual International Symposium on
  Personal, Indoor, and Mobile Radio Communications (PIMRC)}.\hskip 1em plus
  0.5em minus 0.4em\relax IEEE, 2017, pp. 1--5.

\end{thebibliography}

	\end{document}